\documentclass[10pt]{article}
\usepackage{amsthm, amsfonts}
\usepackage{amsmath}
\usepackage{geometry}

\def\Si{\Sigma}
\def\si{\sigma}
\def\lb{\lambda}
\def\Op{\mathfrak{Op}}
\def\hb{\hbar}
\def\H{\mathcal{H}}
\def\de{\mathrm{d}}

\def\AA{\mathfrak A}
\def\A{{\mathcal A}}

\def\R{\mathbb R}

\def\sp{\mathop{\mathrm{sp}}\nolimits}

\newtheorem{proposition}{Proposition}[section]
\newtheorem{Theorem}{Theorem}[section]
\newtheorem{Remark}{Remark}[section]
\newtheorem{Definition}{Definition}[section]
\newtheorem{Lemma}{Lemma}[section]
\newtheorem{Corollary}{Corollary}[section]
\begin{document}
\title{Canonical Quantization of Constants of Motion.}

\date{\empty}
\author{Fabi\'an Belmonte}

\maketitle
\textbf{Address}: Departamento de Matem\'aticas, Universidad Cat\'olica del Norte, Angamos 0610, Antofagasta, Chile; \textbf{email: fbelmonte@ucn.cl}
\vspace{1cm}

\begin{abstract}
We develop a quantization method, that we name {\it decomposable Weyl quantization}, which ensures that the constants of motion of a prescribed finite set of Hamiltonians are preserved by the quantization. 

Our method is based on a structural analogy between the notions of reduction of the classical phase space and diagonalization of selfadjoint operators. We obtain the spectral decomposition of the emerging quantum constants of motion directly from the quantization process.

If a specific quantization is given, we expect that it preserves constants of motion exactly when it coincides with decomposable Weyl quantization on the algebra of constants of motion. We obtain a characterization of when such property holds in terms of the Wigner transforms involved. We also explain how our construction can be applied to spectral theory.    

Moreover, we discuss how our method opens up new perspectives in formal deformation quantization and geometric quantization. 

\end{abstract}


\section{Introduction}
The concept of quantization emerged from the identification of the analogies between the mathematical description of classical and quantum mechanics. The initial idea was to find a way to transform classical observables (i.e. smooth functions on the phase space) into quantum observables (i.e. selfadjoint operators on a Hilbert space) in a physically meaningful manner. The transformation should depend on a parameter $\hb$, interpreted as Planck's constant, and we should be able to recover, from the limit $\hb\to 0$, most of the classical objects with a quantum counterpart that depends on the observables (Poisson bracket, Hamiltonian flows etc.).

The latter is known as canonical quantization and the most important example of it is the Weyl quantization \cite{We, Gro} (see \cite{TE1} for a quite complete review of the different approaches to quantization). 

This work is motivated by the identification of an analogy that we consider important and, perhaps surprisingly, has not been considered in the literature of canonical quantization. Namely, the mathematical descriptions of the notion of constant of motion in both classical and quantum mechanics. More precisely, in this work we shall address the question of whether it is possible to quantize classical constants of motion into quantum constants of motion. To this end, we will exploit two features of symplectic geometry and operator theory that can be considered analogous. The first is that the process of reduction of a phase space (in symplectic geometry) resembles the process of the diagonalization of selfadjoint operator (in operator theory), and the second is that these processes allow to decompose the constants of motion accordingly in their respective descriptions.

To the state of knowledge of the authors, the quantization of constants of motion has not been studied in the context of formal deformation quantization or strict deformation quantization, but has been addressed in the context of geometric quantization during the last decades  \cite{GS,Go}. However, the mathematical assumptions usually imposed in geometric quantization are too restrictive to study the mapping of some important constants of motion. For example, the constants of motion in the canonical phase space $\R^{2n}$ for the position Hamiltonians, momentum Hamiltonians, and the free Hamiltonian, cannot be studied using current geometric quantization methods. 

Therefore, although our method is developed within the framework of canonical quantization, its results have the potential to contribute to the study of constants of motion in other approaches to quantization.

Let us introduce some important mathematical constructions of classical and quantum theory that will serve as a framework for this article. 

Let $\Si^{2n}$ be a real symplectic manifold and denote by $\{\cdot,\cdot\}$ the corresponding Poisson bracket on $C^\infty(\Si)$. Also let $h_1,\cdots, h_k\in C^\infty(\Si)$ be a finite family of complete real Hamiltonians (i.e. the flow along the corresponding Hamiltonian vector field is defined on $\R$) such that $\{h_i,h_j\}=0$ for $1\leq i,j\leq k<n$ (simultaneous Hamiltonians), and denote by $\Phi^1,\cdots,\Phi^k$ their respective flows.

Also let $J=(h_1,\cdots,h_k):\Si\to\R^k$ and define $\Phi_{t_1,\cdots,t_k}:=\Phi^1_{t_1}\circ\cdots\circ\Phi^k_{t_k}$. Since the flows $\Phi^1,\cdots,\Phi^k$ pairwise commute, $\Phi$ defines a Hamiltonian action of $\R^k$ on $\Si$. Then, for each regular $\lb\in J(\Si)$, the constant energy level submanifold $\hat{\Si}_\lb:=J^{-1}(\lb)\subseteq \Si$ is invariant under $\Phi$, and it turns out that, the orbits space $\Si_\lb:=\hat{\Si}_\lb/\Phi$ is a symplectic manifold of dimension $2n-2k$ (endowed with the symplectic form given by Marsden-Weinstein-Meyer reduction considering $J$ as the required equivariant moment map \cite{AM,MW,Me,M}; our particular case is sometimes called Jacobi-Liouville theorem). 
 
On the quantum side (see \cite{D} or \cite{BS} for details), let $H_1,\cdots, H_k$ be a finite family of pairwise strongly commuting selfadjoint operators on a Hilbert space $\H$ and $\sp (H_1,\cdots,H_k)$ its joint spectrum. Then there is a unique Borel measure $\eta$ (up to equivalence) on $\sp (H_1,\cdots,H_k)$, a unique measurable field of Hilbert spaces $\{\sp (H_1,\cdots,H_k)\ni\lb\to\H(\lb)\}$ (up to unitary equivalence on $\eta$-almost every fiber), and a unique unitary operator $T:\H\to\int_{\sp(H_1,\cdots,H_k)}^{\oplus}\H(\lambda)\de\eta(\lambda)$ such that
$$
[T\varphi(H_1,\cdots H_k) u](\lambda)=\varphi(\lb)(Tu)(\lambda)\,\,\,\forall u\in \emph{Dom}(\varphi(H_1,\cdots, H_k)).
$$
where $\varphi$ is any Borel function on $\sp(H_1,\cdots,H_k)$ and $\varphi(H_1,\cdots,H_k)$ denotes the corresponding operator given by the functional calculus. In particular we have that
$$
[TH_j u](\lambda)=\lambda_j(Tu)(\lambda)\,\,\,\forall u\in \emph{Dom}(H_j).
$$
and
\begin{equation}\label{UD}
[Te^{i(t_1H_1+\cdots t_k H_k)} u](\lambda)=e^{i(t_1\lb_1+\cdots +t_k\lb_k)}(Tu)(\lambda)\,\,\,\forall u\in \H, (t_1,\cdots,t_k)\in\R^k.
\end{equation}

$T$ is called the simultaneous diagonalization of the family $H_1,\cdots, H_k$.

The analogy between both constructions is evident: Both constructions are meant to make each Hamiltonian constant on each fiber and the corresponding dynamic trivial. So, heuristically, $\H(\lb)$ is the quantum counterpart of $\Si_\lb$. Interestingly, the latter analogy is somehow a particular case of how constants of motion are represented in classical and quantum theories:

A classical observable $f\in C^\infty(\Si)$ is a constant of motion if $\{f,h_j\} = 0$ for each $j$. Leibniz's rule and Jacobi identity show that the set $\A$ of all constants of motion forms a Poisson subalgebra of $C^\infty(\Si)$. It is easy to show that $f\in\A$ if and only if $f\circ\Phi_t=f$, for each $t\in\R^k$. Let $\pi_\lb:\hat\Si_\lb\to\Si_\lb$ be the quotient map. Thus, for each $f\in\A$, we   
 can consider the field of functions $f_\lb\in C^\infty(\Si_\lb)$ given by 
$$
f_\lb(\pi_\lb(\sigma))=f(\sigma),
$$
where $\sigma\in\hat\Si_\lb$. In particular, we can consider the flow $\Phi_t[f_\lb]$ of $f_\lb$ in $\Si_\lb$. It is not difficult to show that  $\Phi_t[f_\lb]\circ \pi_\lb=\pi_\lb\circ\Phi_t[f]$, where $\Phi_t[f]$ is the flow of $f$ in $\Si$. This a particular case of theorem 4.3.5 in \cite{AM}.

Clearly, if $\varphi\in C^\infty(\R^k)$, then $f=\varphi\circ J$ is a constant of motion. Since $f_\lb$ is constant on $\Si_\lb$, the flow $\Phi_t[f_\lb]$ is trivial  in this case.  

Similarly, if $\psi\in C^\infty(\R)$ and $f\in\A$, then $\psi\circ f$ is also a constant of motion and $\Phi_t[\psi\circ f](\si)=\Phi_{t\psi'[f(\si)]}[f](\si)$. 

A quantum observable, i.e. a selfadjoint operator $F$, is a quantum constant of motion, if $F$ strongly commutes with each $H_j$. It is well known that $F$ is a constant of motion if and only if it admits a decomposition through $T$ \cite{D,BS}, i.e. there is a measurable field of selfadjoint operators $\{\sp(H_1,\cdots,H_k)\ni\lb\to F_\lb\}$ such that $[TFu](\lb)=F_\lb[Tu(\lb)]$. Such field of operators is the quantum counterpart of the field of classical observables $f_\lb$ above.   

It is easy to prove that, if $F$ is a quantum constant of motion, then $\psi(F)$ is also a constant of motion and $\psi(F)_\lb=\psi(F_\lb)$, where $\psi$ is any Borel function on $\R$. In particular, the decomposition of the quantum dynamics defined by $e^{itF}$ is given by $e^{itF_\lb}$. The latter is the quantum version of the particular case of theorem 4.3.5 in \cite{AM} described above in the classical context.     

Notice that the set $\AA$ of bounded quantum constants of motion is a von Neumann algebra isomorphic with the commutant of $L^\infty(\sp(H_1,\cdots,H_k),\eta)$.

To complete the analogy with the statements about classical constants of motion above, notice that we know that $\varphi(H_1,\cdots,H_k)$ is a quantum constant of motion for any Borel function $\varphi$ on $\R^k$ and $\varphi(H_1,\cdots,H_k)_\lb$ is the operator given by multiplication by the constant $\varphi(\lb)$. In particular, the dynamic of $\varphi(H_1,\cdots,H_k)_\lb$ is trivial (just as in the classical case).

In order to relate the analogies stated above to canonical quantization, note that besides sending functions into operators, a canonical quantization might be required to satisfy some extra conditions such as the preservation of position and momentum by the quantization map.

Under the light of the analogy presented above, our aim is to introduce a canonical quantization that is meant to satisfy a new requirement: Given a  set of $k$ commuting classical observables (simultaneous Hamiltonians, with $k<n$) with a corresponding set of $k$ commuting quantum observable, the quantization must be defined over the algebra of classical constants of motion, and must send this algebra into quantum constants of motion.

The reason why we pursue this goal is that the analogies developed above suggest that the notion of constant of motion is in some sense independent of how we decide to represent a physical system (either with classical mechanics or with quantum mechanics), so  we consider necessary to develop
a quantization preserving constants of motion.

The analogies stated above suggests that we can assume the following two properties for the canonical quantization of the constants of motion of our initial system:  i) if $h_j$ is quantized as $H_j$, then $\sp(H_1,\cdots,H_k)\subseteq J(\Si)\backslash\mathcal I$ $\eta$-a.e., where $\mathcal I$ is the set of singular values of $J$, and ii) $\Si_\lb$ can be quantized on $\H(\lb)$. Another justification for these assumptions can be deduced from \cite{Lari,Lacs}, where Rieffel induction is proposed as the quantum counterpart of Marsden-Weinstein-Meyer reduction.

The relation between the range of $J$ and the joint spectrum is a delicated issue which we shall not discuss here. However, we must mention that the inclusion $\sp(H_1,\cdots,H_k)\subseteq J(\Si)\backslash\mathcal I$ seems to be the main obstruction to make our forthcoming construction to work generically. Notice that in the semi-classical limit this problem vanishes. For details we refer to \cite{PN,PPN}, where it was shown that, under some suitable conditions and in some sense, the convex hull of the range of $J$ in the semi-classical limit coincide with the joint spectrum. For $k=1$, we also recommend \cite{AMP}.

Under our assumptions, denote  by $\Op^\lb_\hb$ the quantization mapping functions on $\Si_\lb$ into selfadjoint operators on $\H(\lb)$. Let $f\in\A$ and assume that $\Op_\hb^\lb$ is defined on $f_\lb$ for almost every $\lb$; thus we can essentially define the field of operators $\{ \sp(H_1,\cdots,H_k)\ni\lb\to \Op_\hb^\lb(f_\lb)\}$. Therefore, we can consider the operator $\int_{  \sp(H_1,\cdots,H_k)}^{\oplus}\Op_\hb^\lb(f_\lb)\de\eta(\lambda)$ defined fiberwise on a suitable domain in the Hilbert space
$\int_{ \sp(H_1,\cdots,H_k)}^{\oplus}\H(\lb)\de\eta(\lambda)$.

We then define the decomposable Weyl quantization $\Op^d_\hb$ (informal definition) as
$$
\Op^d_\hb(f):=T^*\left[\int_{ \sp(H_1,\cdots,H_k)}^{\oplus}\Op_\hb^\lb(f_\lb)\de\eta(\lambda)\right]T.
$$

Notice that by construction, $\Op^d_\hb(\varphi\circ J)=\varphi(H_1,\cdots,H_k)$, in particular $\Op^d_\hb(h_j)=H_j$.

In order to study $\Op^d_\hb$, we will need to endow the spaces involved in our framework with volume forms preserving somehow the relation between them. For instance, we will show that there is a natural volume form $\zeta$ on the full orbit space $P:=\Si/\Phi$ and a $k$-form $\alpha$ on $\Si$ such that $\alpha\wedge\pi^*\zeta$ is the Liouville form on $\Si$, where $\pi:\Si\to P$ is the quotient map. Moreover, the restriction of $\alpha$ to each orbit $[\si]$ defines a volume form invariant by $\Phi|_{[\si]}$. The list of properties satisfied by those forms that we would like to emphasize are summarized in theorem \ref{main} and corollary \ref{Weins}.

It is desirable to find the conditions under which a Wigner transform can be defined for our new quantization $\Op^d_\hb$. Assume that each $\Op^\lb_\hb$ admits a Wigner transform $W^\lb_\hb$. That is, for any $u,v\in\H(\lb)$ there is a function $W_{\hb}^\lb(u,v)$ on $\Si_\lambda$ such that
\begin{equation}\label{wig}
  \langle\Op^\lb_\hb(f)u,v\rangle=\int_{\Si_\lb}f W_\hb^\lb(u,v)\,\zeta_\lb,  
\end{equation}

where $\zeta_\lb$ is the Liouville volume form on $\Si_\lb$. It is well known that $P$ is a Poisson manifold and its symplectic leaves are the submanifolds $\Si_\lb$. Smooth functions on $P$ correspond to constants of motion through the pull back of the canonical projection $\pi:\Si\to P$. Assume that the family $H_1,\cdots, H_k$ have absolutely continuous joint spectrum and define the volume form $\zeta_H:=(\frac{\de \eta}{\de\lb}\circ J)\zeta$ on $P$, where $\frac{\de \eta}{\de\lb}$ is the Radon-Nikodym derivative of $\eta$ with respect to the Lebesgue measure. Then the following formula holds
$$
\int_P f\,\zeta_H= \int_{J(\Si)}\int_{\Si_\lb}f\de\zeta_\lb\de\eta(\lb),
$$
for any integrable function $f$ on $P$. For each $u,v\in\H$, 
let $W_\hb^d(u,v)([\si]):=W_\hb^{J(\si)}(u_{J(\si)},v_{J(\si)})([\si])$, where $u_\lb:=Tu(\lb)$; in other words we define $W_\hb^d$ gluing on $P$ the values of each Wigner transform $W_\hb^\lb$. We call $W_\hb^d$ the decomposable Wigner transform. Indeed, assuming that $W_\hb^d (u,v)$ is a Borel function on $P$ we have that
\begin{equation}
\langle\Op^d_\hb(f)u,v\rangle=\int_{P}f_d W_\hb^d(u,v)\,\zeta_H,
\label{decomp-wigner}    
\end{equation}

where $f_d$ denotes the function on $P$ corresponding to the constant of motion $f$. Formula~\ref{decomp-wigner} can be useful to study $\Op_\hb^d$. For instance, using the Riesz representation theorem, we can identify the domain of $\Op_\hb^d(f)$ and obtain a basic characterization of the conditions under which $\Op_\hb^d(f)$ defines a bounded operator. Another possible approach to study $\Op_\hb^d(f)$ is to look for a suitable subspace $Q_f\subseteq\H$ such that the right hand side of Eq.~\ref{decomp-wigner} defines a quadratic form. Additionally, if $W_\hb^d(u,v)$ is smooth and has compact support, Eq.~\ref{decomp-wigner} also allows to define $\langle\Op_\hb^d(f)u,v\rangle$ even if $f$ is a distribution on $P$.

Back to the quantization problem, we have by construction that $\Op^d_\hb$ satisfies
$$
\{f,h_j\}=0\Rightarrow [\Op^d_\hb(f),H_j]=0.
$$

However, we do not know if this property is satisfied by other quantizations. For instance, if $\Si=\R^{2n}$ and $\Op_\hb$ denotes the canonical Weyl quantization, then we only know that
$$
\{f,h_j\}=0\Rightarrow [\Op_\hb(f),H_j]=\mathcal O(\hb^2).
$$

In addition, the Groenewold-Van Hove's no go theorem \cite{Gro,vHo,Fol} implies that in general we do not have that $\Op_\hb(\{f,g\})=\frac{1}{\hb}[\Op_\hb(f),\Op_\hb(g)]$. Thus, 
this theorem suggests that it might not be always the case that quantizations map all classical commuting observables into quantum commuting observables; in other words it seems that a given quantization might fail to preserve constants of motion generically. 

We will use decomposable Weyl quantization as a tool to approach the questions of when and how a given quantization $\Op_\hb$ preserves the constants of motion of certain set of classical commuting observables.

In the case that $\Op_\hb$ does preserve constants of motion for certain $h_1,\cdots, h_k$, 
knowing that $f\in\A$ and $\Op_\hb(f)$ is selfadjoint, do not necessarily allow to obtain the explicit expression of the field of operators $\{\sp(H_1,\cdots,H_k)\ni\lb\to[\Op_\hb(f)]_\lb\}$ that decompose $\Op_\hb(f)$ through $T$. Interestingly, the construction of decomposable Weyl quantization might give the answer to such problem: We could expect that under reasonable conditions we should have that $[\Op_\hb(f)]_\lb=\Op^\lb_\hb(f_\lb)$, or equivalently $\Op_\hb(f)=\Op^d_\hb(f)$; in such case we say that we have {\it commutation of canonical quantization and reduction} on $f$ (CQR for short).

Notice that finding $[\Op_\hb(f)]_\lb$ is important for spectral theory because the spectrum of $\Op_\hb(f)$ is the closure of the union over $\lb\in \sp(H_1,\cdots,H_k)$ of the spectrum of each $[\Op_\hb(f)]_\lb$.

Heuristically, $\Op_\hb$ preserves constants of motion in the semi-classical limit. So, one could expect that $\Op_\hb^d$ and $\Op_\hb$ coincide in the limit $\hb\to 0$. In such case, we say that we have {\it semi-classical commutation of quantization and reduction} (SCQR for short). 

We will characterize commutation of canonical quantization and reduction in terms of the Wigner transforms involved, when the family $H_1,\cdots, H_k$ have absolutely continuous joint spectrum. Formally, our result states that we will have CQR on every constant of motion iff $(\frac{\de\eta}{\de\lb}\circ J)(\si) W_\hb^d(u,v)([\si])=\int_{[\si]}W_\hb(u,v)\alpha_{[\si]}$ for each $u,v\in\H$ and $\si\in\Si$, where $W_\hb$ is the Wigner transform of $\Op_\hb$ and $\alpha_{[\si]}$ is the restriction of $\alpha$ to $[\si]$ (theorem \ref{char}). Another characterization can be obtained in terms of the average $\int_{\R^k} W_\hb(u,v)\circ \Phi_t(\si)\de t$ using corollary \ref{Weins}. 

In order to show the usefulness of our construction we develop  examples where the conditions to build $\Op_\hb^d$ are satisfied, and we show how to derive explicit expressions of $\Si_\lb$ and $\H(\lb)$.

Namely, consider $\Si=\R^{2n}$ endowed with the canonical symplectic form and $h_j(x,\xi)=\phi_j(x)$, where each $\phi_j$ belongs to $C^\infty(\R^n)$. Also let $\tilde J=(\phi_1,\cdots\phi_k)$ and $\Omega_\lb=(\tilde J)^{-1}(\lb)$. In this case, clearly $\hat\Si_\lb=\Omega_\lb\times\R^n$, and we are going to prove that $\Si_\lb\cong T^*\Omega_\lb$ (theorem \ref{sym}). On the quantum side, we quantize each $h_j$ as $H_j=\phi_j(Q)$, i.e. $H_j$ is the multiplication operator given by $\phi_j$ acting in $L^2(\R^n)$. Then $\sp(H_1,\cdots,H_k)=\overline{ J(\R^{2n})}$. We are going to show that $\eta(A)=\mathcal L^k(A\cap  J(\R^{2n}))$, where $\mathcal L^k$ is the Lebesgue measure on $\R^k$ (lemma \ref{mea}). It will become clear that $\{ J(\R^{2n})\backslash\mathcal I\ni\lb\to L^2(\Omega_\lb)\}$ forms a $\eta$-measurable field of Hilbert spaces and, using the co-area formula, we are going to find $T$ explicitly with values in the direct integral of $L^2(\Omega_\lb)$ respect to $\eta$ (theorem \ref{DI}). 

We will show that the same statements of the previous paragraph hold for $h_j=\tilde h_j\circ S$, when $S$ is a linear symplectomorphism and $\tilde h_j(x,\xi)=\phi_j(x)$, where $\phi_j\in C^\infty(\R^n)$. The main tool to prove the latter result is the metaplectic representation. Notice that $f$ is constant of motion for the family $\tilde h_j$ iff $f\circ S$ is a constant of motion for the family $h_j$. The details and proofs of the statements above concerning reduction and diagonalization can be found in section \ref{mec}. 

Applying the criteria given by theorem \ref{char}, we show that it is impossible to obtain CQR on {\it every} constant of motion for our examples, unless each $h_j$ is linear \ref{CLW}. The main obstacle is that $\tilde W(u,u)=\int_{[\si]}W_\hb(u,u)\alpha_{[\si]}$, with $[\si]\in\Si_\lb$, does not depend only on $u|_{\Omega_\lb}$, so it does not define a Wigner transform. However, perturbing in a natural way $\tilde W$ we do obtain a Wigner transform that was already defined by Landsman (\cite{LaYM} or\cite{La} II.3). Therefore, we can construct $\Op^d_\hb$ in the examples described above and we expect SCQR. Furthermore, we show $CQR$ for an important class of constant of motion \ref{ecqr}. We also include a discussion concerning a wider class of constant motion of a very important examples where we still expect CQR.

We finish this article with section \ref{OQ}, where we explain very briefly how our problem has been studied from the perspective of geometric quantization, what are the differences with our approach, and what could be done within the framework of deformation quantization. 

\section{Volume forms and integration identities.}

We begin providing some integration properties over the spaces involved in our framework. In particular, we shall endow the Poisson manifold of orbits $P:=\Si/\Phi$ with a non-degenerate volume form $\zeta$ suitable for our purpose. Theorem \ref{main} and its corollary \ref{Weins} are interesting by themselves within the framework of symplectic geometry.
 
Recall that the symplectic leaves of $P$ are precisely the submanifolds $\Si_\lb$ and $C^\infty(P)$ can be identified with the algebra of constants of motion $\A$. For each $\si\in\Si$, we will denote by $[\si]$ the corresponding orbit. We shall denote by $\Theta$ and $\Theta_\lb$ the symplectic forms on $\Si$ and $\Si_\lb$ respectively, and by $\omega$  and $\zeta_\lb$ the corresponding (Liouville) volume forms. We will also denote by $X_1,\cdots, X_k$ the Hamiltonian vector fields corresponding to $h_1,\cdots, h_k$ respectively. 
\begin{Lemma}\label{mainL}
Let $M^m$ and $N^l$ be smooth manifolds endowed with volume forms $\nu$ and $\mu$ respectively such that $m>l$ and $\mu$ is non-degenerate. Also let $A:M\to N$ be a smooth regular function. There is a $(m-l)$-form $\gamma$ on $M$ such that $\nu=\gamma\wedge A^*\mu$. In particular, for any integrable function $f$ on $M$ we have that
$$
\int_M f\,\nu=\int_{A(M)}\left(\int_{A^{-1}(p)}f \,\gamma_p\right)\mu(p),
$$
where $\gamma_p=i_p^*\gamma$ and $i_p:A^{-1}(p)\to M$ is the inclusion map.
\end{Lemma}
\begin{proof}
Let $(U,x_1,\cdots, x_m)$ be a chart on $M$ and $(V,p_1,\cdots,p_l)$ be a chart on $N$ such that $A(U)\subseteq V$. The regularity of $A$ implies that there is a $l\times l$ submatrix of the Jacobian matrix of $A$ with non-trivial determinant, and the implicit function theorem implies that we can obtain another chart on $M$ by removing the $l$ coordinates corresponding to such submatrix from $(U,x_1,\cdots, x_m)$ and replacing them by $y_j:=p_j\circ A$. Notice that, for each $p\in V$, only removing the latter $l$ coordinates (without replacing them) also defines a chart $(U\cap A^{-1}(p),x_{i_1},\cdots,x_{i_{m-l}})$ on $A^{-1}(p)$.  

For simplicity, assume that $det(\frac{\partial y_j}{\partial x_i})_{i,j\leq l}\neq 0$.
Let $g$ and $\varphi$ smooth functions on $U$ and $V$ respectively such that $\nu|_U=g\de x_{l+1}\wedge\cdots\wedge\de x_{m-l}\de y_1\wedge\cdots\wedge\de y_l$ and $\mu|_V=\varphi\de p_1\wedge\cdots\wedge\de p_l$. Since $\varphi(p)\neq 0\forall p\in V$, we have that
$$
\nu|_U=\frac{g}{\varphi\circ A}\de x_{l+1}\wedge\cdots\wedge\de x_{m-l}\wedge A^*(\mu)|_U.
$$
Therefore, taking $\gamma|_U=\frac{g}{\varphi\circ A}\de x_{l+1}\wedge\cdots\wedge\de x_{m-l}$ and using a partition of unity, we obtain our result
\end{proof}
\begin{Theorem}\label{main}
There are a non-degenerate volume form $\zeta$ on $P$, a $(2n-2k)$-form $\beta$ on $P$ and a $k$-form $\alpha$ on $\Si$ such that: 
\begin{enumerate}
\item[a)] $\zeta=\beta\wedge\de \tilde h$, where $\de \tilde h=\de \tilde h_1\wedge\cdots\wedge\de\tilde h_k$ and $\tilde h_j\in C^\infty(P)$ is given by $\tilde h_j([\si])=h_j(\si)$. Moreover, $ j_\lb^*\beta=\zeta_\lb$, where $j_\lb:\Si_\lb\to P$ is the inclusion map.

\item [b)] $\omega=\alpha\wedge\pi^*\zeta$, where $\pi:\Si\to P$ is the quotient map. Moreover $i(X)\omega=\alpha(X_1,\cdots,X_k)\pi^*\zeta$ and $i^*_{[\si]}\alpha$ is invariant by $\Phi_t|_{[\si]}$, where $i(X)\omega$ is the consecutive interior product of $\omega$ by all the vector fields $X_1,\cdots X_k$ and $i_{[\si]}:[\si]\to\Si$ is the inclusion map. If $\alpha'$ is another $k$-form such that $\omega=\alpha'\wedge\pi^*\zeta$, then $i^*_{[\si]}\alpha'=i^*_{[\si]}\alpha$.

\item [c)] For every integrable function $f$ on $P$, we have that 
$$
\int_Pf\zeta= \int_{J(\Si)}\left(\int_{\Si_\lb}f \,\zeta_\lb\right)\de\lb.
$$
\item [d)] For every integrable function $f$ on $\Si$, we have that  
$$
\int_\Si f\,\omega =\int_P\left(\int_{[\si]}f\, i^*_{[\si]}\alpha\right)\zeta([\si]).
$$
\end{enumerate}
\end{Theorem}
\begin{proof}
We begin endowing $P$ with a non-degenerate volume form $\zeta_0$ (i.e. we will show that $P$ is orientable). Since the symplectic form $\omega$ is non-degenerate and the co-vectors $\de h_1(\si),\cdots,\de h_k(\si)$ are linearly independent, the vectors $X_1(\si), \cdots, X_k(\si)$ are also linearly independent, for every $\si\in\Si$. Moreover, since $\pi$ is a submersion, $\text{dim}(\text{Ker}(D\pi_\si))=k$, for every $\si\in\Si$. Therefore, $\text{Ker}(D\pi_\si)$ is the space generated by the vectors $X_1(\si), \cdots, X_k(\si)$, for every $\si\in\Si$. The latter fact allow us to define, for each $\tilde v_1,\cdots,\tilde v_{2n-k}\in T_{[\si]}P$, 
$$
(\zeta_0)_{[\si]}(\tilde v_1,\cdots,\tilde v_{2n-k})=(i(X) \omega)_\si(v_1,\cdots, v_{2n-k}),
$$
where $ v_1,\cdots, v_{2n-k}\in T_{\si}\Si$ are such that $D\pi_\si(v_j)=\tilde v_j$, for all $j=1,\cdots, 2n-k$. Since $\omega$ is non-degenerate and the vectors $X_1(\si), \cdots, X_k(\si)$ are linearly independent, $\zeta_0$ is also non-degenerate.

Define $ J^d:P\to\R^k$ by $J^d([\si])=J(\si)$, for each $\si\in\Si$. Since $J$ is constant on each orbit, $J^d$ is well defined, smooth and $J^d\circ \pi=J$. Moreover, $D J^d(D\pi(v))=DJ(v)$, thus $J^d$ is also regular. Therefore, the previous lemma implies that there is a $(2n-2k)$-form $\beta_0$ on $P$ such that $\zeta_0=\beta_0\wedge\de\tilde h$. Since $j_\lb^*\beta_0$ is a volume form on $\Si_\lb$, there is a smooth function $g$ on $P$ such that $g j_\lb^*\beta_0=\zeta_\lb$. Thus, defining $\zeta=g\,\zeta_0$ and $\beta=g\,\beta_0$, we obtain $a)$ and $c)$. 

In order to obtain $\alpha$ satisfying $b)$ and $d)$, we apply the previous lemma over $\pi:(\Si,\omega)\to (P,\zeta)$. 

Let $g$ be a function on $P$ such that $\zeta=g\zeta_0$.  Since $X_1,\cdots,X_k$ are linearly independent, there are vectors $v_1,\cdots,v_{2n-k}$ such that $\omega(X_1,\cdots, X_k,v_1,\cdots,v_{2n-k})\neq 0$. Notice that $i(X)\omega$ vanish whenever we evaluate it on any $(2n-k)$ vectors containing some $X_j$. Moreover, by construction we know that  $g\alpha\wedge i(X)\omega=\omega$, and evaluating both sides of this equation on $(X_1,\cdots, X_k,v_1,\cdots,v_{2n-k})$, we obtain $g\,\alpha(X_1,\cdots,X_k)=1$, which shows the second equation in $b)$.

Let $(V,p_1,\cdots,p_{2n-k})$ be a chart on $P$. Recall from the proof of the previous lemma that, we can always take a chart on $\Si$ such that each $p_j\circ \pi$ is a coordinate. Also notice that, since $\pi^*\zeta$ is invariant by $\Phi_t$, we have that
$$
0=L(X)\omega=(L(X)\alpha)\wedge \pi^*\zeta,
$$
where $L(X)=L(X_1)\cdots L(X_k)$ and $L(X_j)$ is the Lie derivative with respect to the vector field $X_j$. Moreover, since $\pi^*\zeta=f\de (p_1\circ\pi)\wedge\cdots\wedge \de (p_k\circ\pi)$, for some non-vanishing function $f$, there are $(k-1)$-forms $\tau_1,\cdots,\tau_k$ such that
$$
L(X)\alpha=\sum_j \de (p_j\circ\pi)\wedge\tau_j.  
$$
Since $i^*_{[\si]}\de (p_j\circ\pi)=0$, we have that $L(X)i^*_{[\si]}\alpha=i^*_{[\si]}(L(X)\alpha)=0$, i.e. $i^*_{[\si]}\alpha$ is invariant by $\Phi|_{[\si]}$. Finally, if $\alpha'\wedge\pi^*\zeta=\alpha'\wedge\pi^*\zeta$, then $(\alpha-\alpha')\wedge\pi^*\zeta=0$ and, repeating the same argument than before, we have that
$$
\alpha-\alpha'=\sum_j \de (p_j\circ\pi)\wedge\nu_j,
$$
where $\nu_1,\cdots,\nu_k$ are certain $(k-1)$-forms. Therefore $i^*_{[\si]}(\alpha-\alpha')=0$ and this finish the proof of $b)$.
\end{proof}
\begin{Remark}\label{gamma}
{\rm
Since $\omega=\alpha\wedge\pi^*\beta\wedge \de h$, where $\de h=\de h_1\wedge\cdots\wedge \de h_k$, theorem 3.4.15 in \cite{AM} implies that 
$\gamma_\lb:=i^*_\lb(\alpha\wedge\pi^*\beta)$ is a volume form on $\hat\Si_\lb$ invariant by $\Phi_t$. Therefore, for every integrable function $f$ on $\Si$, we have that  
\begin{equation}\label{gamma}
\int_\Si f\,\omega = \int_{J(\Si)}\left(\int_{\hat\Si_\lb}f\,\gamma_\lb\right)\de\lb
\end{equation}

Since $i_\lb^*\Theta=\pi_\lb^*\Theta_\lb$, we have that $\gamma_\lb=i^*_\lb(\alpha\wedge\Theta^{n-k})$. Moreover, for any integrable function on $\tilde\Si_\lb$, we have that
\begin{equation}\label{zeta}
\int_{\hat\Si_\lb}f\,\gamma_\lb=\int_{\Si_\lb}\left(\int_{[\si]}f\, i^*_{[\si]}\alpha\right)\zeta_\lb.
\end{equation}

}
\end{Remark}
Let $L_{[\si]}$ the stabilizer subgroup of the orbit $[\si]$, i.e. $L_{[\si]}=\{t\in\R^{k}/\Phi_t(\si)=\si\}$, where $\si$ is any element of $[\si]$ (clearly the stabilizer depends only of the orbit). Let $G_{[\si]}=\R^k/L_{[\si]}$. It is well known that the map $A_{\si}: G_{[\si]}\to[\si]$ given by $A_{\si}([t])=\Phi_t(\si)$ defines a diffeomorphism.  Since $G_{[\si]}$ is a locally compact group, it has a unique Haar measure $\mu$ (up to some constant positive factor). If $L_{[\si]}=\{0\}$, we will take $\mu$ to be the Lebesgue measure. If $L_{[\si]}\neq\{0\}$, then $L_{[\si]}$ is a lattice (the regularity of $J$ implies that $L_{[\si]}$ has no copy of $\R$ as a subgroup), $G_{[\si]}$ is a tori and we will take $\mu$ normalized, i.e. $\mu(G_{[\si]})=1$.

\begin{Corollary}\label{Weins}
Let $\alpha_{[\si]}=i^*_{[\si]}\alpha$. For each orbit $[\si]$, there is a constant $a([\si])$, such that for any integrable function $f$ on $[\si]$ we have that
$$
\int_{[\si]}f\,\alpha_{[\si]}=a([\si])\int_{G_{[\si]}} (f\circ \Phi_t)(\si)\,\mu(t).
$$
In particular, if $L_{[\si]}\neq 0$, then $a([\si])$ is the volume of $[\si]$ according to $\alpha_{[\si]}$.
\end{Corollary}
\begin{proof}
Let $\si\in\Si$ and consider. Since $A_\si$ is a diffeomorphism, we have that
$$
\int_{[\si]}f\,\alpha_{[\si]}=\int_{G_{[\si]}} (f\circ \Phi_t)(\si)\,A_\si^*\alpha_{[\si]}(t)
$$

For each $s\in G_{[\si]}$, defines $\Psi_{s}:G_{[\si]}\to G_{[\si]}$ by $\Psi_s(t)=t+s$. It is straightforward to check that $\Phi_s\circ A_{\si}=A_\si\circ\Psi_s$. Thus,
$$
 A_{\si}^*\alpha_{[\si]}=A_{\si}^*\Phi_t^*\alpha_{[\si]}=\Psi_t^*A_{\si}^*\alpha_{[\si]}.
$$
Since the Haar measure $\mu$ is the only measure invariant by translation up to some constant factor, there must be a constant $a([\si])$ such that $A_{\si}^*\alpha_{[\si]}=a([\si])\mu$ and this finishes our proof.
\end{proof}
\section{The decomposable Wigner transform.}

The canonical definition of Weyl quantization through a kernel leads to some technical difficulties which can be overcome introducing the so called Wigner transform \cite{Fol,La}. It is also one the main object in the so called phase space description of quantum mechanics. We shall define an analogue object for our decomposable Weyl quantization. 

In this section we will not need to consider the role of Planck's constant $\hb$, and for simplicity, we shall not include it in the notation.

Using theorem \ref{main} we can endow $P$ with a volume form suitable for our quantization. 
\begin{proposition}
\label{Pintegration}
Assume $H_1,\cdots H_k$ has absolutely continuous joint spectrum and define the (measurable) volume form $\zeta_H=(\frac{\de\eta}{\de\lb}\circ J )\zeta$ on $P$, where $\frac{\de\eta}{\de\lb}$ is the Radon-Nikodym derivative of $\eta$ with respect to the Lebesgue measure. For any $\zeta_H$-integrable function $f$ on $P$ we have that
\begin{equation}\label{zetaH}
\int_Pf\zeta_H= \int_{\sp(H_1,\cdots,H_k)}\left(\int_{\Si_\lb}f\zeta_\lb\right)\de\eta(\lb).
\end{equation}
If $\eta$ is equivalent to the Lebesgue measure on $J(\Si)$, then $\zeta_H$ is non-degenerate.
\end{proposition}

The domain of $\Op^d(f)$ is quite difficult to analyze in its current form. Indeed, $u\in Dom(\Op^d(f))$ if: i) $u_\lb\in Dom(\Op^\lb(f_\lb))$, where $u_\lb=Tu(\lb)$ and ii) the section $\{\sp(H_1,\cdots,H_k)\ni\lb\to\Op^\lb(f_\lb)u_\lb\}$ belongs to the direct integral $\int_{\sp(H_1,\cdots,H_k)}^\oplus\H(\lb)\de\eta(\lb)$.   
\begin{Corollary}
Let $u\in\H$ such that $u_\lb\in Dom(\Op^\lb(f_\lb))$ and the section $\{J(\Si)\ni\lb\to\Op^\lb(f_\lb)u_\lb\}$ belongs to the direct integral $\int_{J(\Si)}^\oplus\H(\lb)\de\eta(\lb)$. Let $W^d(u,v)([\si])=W^{J(\si)}(u_{J(\si)},v_{J(\si)})([\si])$, where $[\si]\in P$ denotes de orbit of $\si\in\Si$. Assume that $W^d (u,v)$ is a Borel function on $P$. Thus,
$$
\langle\Op^d(f)u,v\rangle=\int_{P}f_d W^d(u,v) \zeta_H,
$$
where $f_d$ denotes the function on $P$ corresponding to the constant of motion $f$.
\end{Corollary}
\begin{proof}
$$
\langle\Op^d(f)u,v\rangle=\langle \int^\oplus_{\sp(H_1,\cdots,H_k)}\Op^\lb(f_\lb)\de\eta(\lb) T u,T v\rangle=\int_{\sp(H_1,\cdots,H_k)}\langle\Op^\lb(f_\lb)u_\lb,v_\lb\rangle\de\eta(\lb)=
$$
$$
=\int_{\sp(H_1,\cdots,H_k)}\int_{\Si_\lb}f_\lb W^\lb(u_\lb,v_\lb)\zeta_\lb\de\eta(\lb)=\int_{P}f_d W^d(u,v)\,\zeta_H,
$$
\end{proof}
The previous corollary allow us to define $\Op^d(f)$ using the Riesz representation theorem.
\begin{Definition}
We say that $u$ belongs to $Dom(\Op^d(f))$ if there is a dense subspace $\H_u$ in $\H$ such that 
\begin{enumerate}
    \item[i)] $W^d(u,v)$ is Borel measurable in $P$, for any $v\in\H_u$.
    \item[ii)] the map $\H_u\ni v\to\int_P f_d W^d(u,v)\,\zeta_H$ defines a continuous linear functional on $\H_u$.
\end{enumerate}
If $u\in Dom(\Op^d(f))$, we define $\Op^d(f)u$ as the only vector $w\in\H$ such that, for any $v\in\H_u$ we have that  
$$
\langle w,v\rangle=\int_{P}f_d W^d(u,v)\,\zeta_H,
$$
\end{Definition}
Notice that, by the Riesz representation theorem $||\Op^d(f)u||=\sup\{\left|\int_P f_d W^d(u,v)\,\zeta_H\right|/||v||\leq 1\}$. Therefore $\Op^d(f)$ is bounded if and only if 
$$
\sup\left\{\left|\int_P f_d W^d(u,v)\,\zeta_H\right|/||u||,||v||\leq 1\right\}<\infty
$$
In particular, if $f\in L^1(P,\zeta)$, $W^d(u,v)$ is Borel measurable and $||W^d(u,v)||_\infty\leq C$ for $||u||,||v||\leq 1$, then $\Op^d(f)$ is bounded and $||\Op^d(f)||\leq C||f||_1$.

We expect that the family of Wigner transforms $\{W^\lb\}_{\lb\in J(\Si)}$ satisfies some of the properties that canonical Wigner transform does. For instance, we expect that, for each $u_\lb,v_\lb\in \H_\lb$, $W^\lb(u_\lb,v_\lb)\in C_0(\Si_\lb)$ and there is $C>0$ such that
\begin{equation}\label{ineq}
||W^\lb(u_\lb,v_\lb)||_\infty\leq C||u_\lb||_\lb\cdot ||v_\lb||_\lb,   
\end{equation}
If such property holds, then $W^d(u,v)$ is continuous, and if in addition\\ $||f||_{1,\infty}:=\sup_{\lb\in \sp(H_1,\cdots,H_k)}{\int_{\Si_\lb}|f_\lb|\zeta_\lb}<\infty$, then
\begin{align*}
 \left|\int_P f_d W^d(u,v)\de\zeta\right|&\leq  \int_{\sp(H_1,\cdots,H_k)}\int_{\Si_\lb}||W^\lb(u_\lb,v_\lb)||_\infty|f_\lb|\de\zeta_\lb\de\lb \\
 & \leq C||f||_{1,\infty}\int_{\sp(H_1,\cdots,H_k)}||u_\lb||||v_\lb||\de\lb\leq C||f||_{1,\infty}||u||||v||.
\end{align*}
Therefore, $\Op^d(f)$ is bounded and $||\Op^d(f)||\leq C||f||_{1,\infty}$.

Another interesting approach to study $\Op^d$ using the decomposable Wigner transform $W^d$ is the construction of quadratics forms. Assume that there is a dense subspace $\H_f$ such that $W^d(u,u)$ is Borel medible and $\int_{P}f_d W^d(u,u)\de\zeta$ is convergent,  for any $u\in\H_f$. Polarization identity implies that $\int_{P}f_d W^d(u,v)\de\zeta$ is convergent,  for any $u,v\in\H_f$, so we can define
$$
q_f(u,v)=\int_{P}f_d W^d(u,v)\de\zeta.
$$
Clearly $q_f$ is a symmetric quadratic form. We will study in future works whether $q_f$ is closed, semi- bounded or positive. 


Notice that the definition of $\Op^d(f)$ make sense if $f_d$ is Borel measurable, the same happens with the definition of the quadratic form $q_f$. Moreover, if we impose some additional conditions on $W^d$ we can give sense to $\Op(f)$ for a wider class of "functions" $f_d$ on $P$. More precisely, if $W^d(u,v)\in C_c^\infty(P)$ then $\int_{P}f_d W^d(u,v)\de\zeta$ makes sense if $f_d$ is a distribution on $P$. Therefore, if there is a locally convex space $S$ continuously embedded in $\H$ such that $W^d(u,v)\in C_c^\infty(P)$ for all $u,v\in S$ and \eqref{ineq} holds, then $\Op^d(f):S\to S'$ defines a continuous operator for any smooth distribution $f$, where $S'$ is the topological dual of $S$ endowed with the weak* topology.

\section{Commutation of Canonical Quantization and Reduction.}\label{DWC}

In this section we will assume that there is a quantization $\Op_\hb$ previously defined on $\Si$ with a suitable Wigner transform $W_\hb$ such that $\Op_\hb(h_j)=H_j$, for all $j=1,\cdots, k$. 

We would like to know when $\Op_\hb$ preserves the constants of motion of $h_1,\cdots h_k$ and when it fails to do so. We propose to use the decomposable Weyl quantization to approach that problem. We also explain how the ideas we develop in the way can be useful for applications in spectral theory.

Recall that the main reason why we constructed $\Op^d_\hb$ was to guarantee that classical constants of motion of $h_1,\cdots, h_k$ are are mapped into quantum constants of motion of $H_1,\cdots, H_k$. It seems that $\Op_\hb$ might not have that property generically; for instance, Groenewold-van Hove theorem \cite{Gro,vHo, Fol} suggest that Weyl quantization does not have that property in general. However, there are some important examples where $\Op_\hb$ preserves constants of motion. For example, if $\Op_\hb$ is Weyl canonical quantization, using the Moyal product expansion, one can prove that $\Op_\hb\{h_j,f\}=\frac{1}{\hb}[H_j,\Op_\hb(f)]$, whenever $h_1,\cdots, h_k$  are polynomials of degree at most 2; in particular Weyl quantization preserves constants of motion on such cases. However, intertwining the Poisson bracket with the commutator seems to be stronger than preserving constants of motion, so we expect that there should be more examples. 

Assume that $\Op_\hb$ preserves constants of motion for some $h_1,\cdots, h_k$ and let $f\in\A$ such that $\Op_\hb(f)$ defines a selfadjoint operator. Thus, there is a field of operators $\{\sp(H_1,\cdots , H_k)\ni\lambda\to[\Op_\hb(f)]_\lb\}$ decomposing $\Op_\hb(f)$ through $T$ (recall that quantum constant of motion can be decomposed); however in principle we do not know how to find explicitly such field. The definition of $\Op^d_\hb(f)$ start with the construction of a field of operators,  one could expect that this field gives the correct answer for such problem, i.e. we might expect that $[\Op_\hb(f)]_\lb=\Op_\hb^\lb(f_\lb)$. In other words, $\Op_\hb^d$ coincides with $\Op_\hb$ when the latter sent classical constants of motion into quantum constants of motion. 

Even if $f\in\A$ and $[\Op_\hb(f),H_j]\neq 0$, we know that $\lim_{\hb\to 0}[\Op_\hb(f),H_j]=0$. This suggest that $\Op^d_\hb(f)$ and  $\Op_\hb(f)$ might coincide in the limit $\hb\to 0$. The latter ideas motivate the following definitions.
\begin{Definition}\label{defCQR}
Let $\Op^d_\hb$ be a decomposable Weyl quantization and $\Op_\hb$ a prescribed canonical quantization. Also, let $f$ be a constant of motion.
\begin{enumerate}
\item[a)]
We say that we have commutation of canonical quantization and reduction on $f$ (CQR) if for each $u\in \text{Dom}(\Op_\hb(f))\cap \text{Dom}(\Op^d_\hb(f))$,
$$
\Op_\hb(f)u=\Op^d_\hb(f)u.
$$
\item[b)]
We say that we have semi-classical commutation of canonical quantization and reduction on $f$ (SCQR)  if for each $u\in \text{Dom}(\Op_\hb(f))\cap \text{Dom}(\Op^d_\hb(f))$,
$$
\lim_{\hb\to 0}\left(\Op_\hb(f)-\Op^d_\hb(f)\right)u=0.
$$
\end{enumerate}

\end{Definition}
Clearly canonical quantization and reduction commute on
$f$ iff, for each $u\in \text{Dom}(\Op_\hb(f))\cap \text{Dom}(\Op^d_\hb(f))$, we have that

\begin{equation*}
[T\Op_\hb(f)u](\lb)=\Op^\lb_\hb(f^\lb)[Tu(\lb)],
\end{equation*}

Notice that by construction, canonical quantization and reduction commutes on $f=h_j$ for each $j=1,\cdots, k$, but this does not necessarily happen for $f=\varphi\circ h_j$, because $\varphi(H_j)$ might not coincide with $\Op(\varphi\circ h_j)$. 
  
\begin{Remark}
{\rm
We borrow the term "commutation of quantization and reduction" from the analogue theory in geometrical quantization; we shall insist that the main results of that theory do not cover our setting.
}
\end{Remark}

\begin{Remark}
{\rm
If $f\in\A$ and $\Op_\hb(f)$ defines a selfadjoint operator, a necessary condition to obtain commutation of quantization and reduction is that $\Op_\hb(f)$ strongly commutes with each $H_j$ (i.e. $\Op_\hb(f)$ is a quantum constant of motion). We do not know if this condition is also sufficient. For instance, we know that Weyl quantization preserves constant of motion of the free Hamiltonian, but we do not know yet for which $f$ we have CQR.
}
\end{Remark}
\begin{Remark}
{\rm
Assume that quantization and reduction commute on $f$. One of the main advantages of such property is the following: if $\Op_\hb(f)$ defines a selfadjoint operator, then we have that
\begin{equation}\label{specdesc}
\sp(\Op_\hb(f))=\overline{\cup_\lb\sp(\Op_\hb^\lb(f^\lb)) }.
\end{equation}
}
\end{Remark}

\begin{Remark}
{\rm
If we have that quantization and reduction commute on $f$ semiclassically, then we would interpret $\Op_\hb^d(f)$ as an effective Hamiltonian for $\Op_\hb(f)$. In this case, we would expect that \eqref{specdesc} holds only in the limit $\hb\to 0$ (we will not explain the meaning of that limit here).  

However, notice that we have not found yet examples where we have that quantization and reduction commute only semi-classically. If there is no example where quantization and reduction commute only semi-classically, it would suggest that $\Op_\hb$ always preserves constants of motion (despite Groenewold-van Hove no go theorem, when $\Op_\hb$ is the canonical Weyl quantization). If there are such examples it would mean that $\Op^d_\hb$ is truly a novel quantization procedure. Both possibilities are remarkable.} 

\end{Remark}

The following result provides necessary and sufficient conditions to guarantee commutation of canonical quantization and reduction on every constant of motion when the family of selfadjoint operators $H_1,\cdots,H_k$ have absolutely continuous joint spectrum.
\begin{Theorem}\label{char}
Assume that $H_1,\cdots,H_k$ have absolutely continuous joint spectrum and let $W_\hb$ be the Wigner transform of $\Op_\hb$. Quantization and reduction commute on every constant of motion if and only if, for any $u,v\in\H$ and $\si\in\Si$, we have that
\begin{equation}\label{cond}
(\frac{\de\eta}{\de\lb}\circ J)(\si)W_\hb^d(u,v)([\si])=\int_{[\si]}W_\hb(u,v)\,i^*_{[\si]}\alpha, 
\end{equation}

or equivalently,
$$
(\frac{\de\eta}{\de\lb}\circ J)(\si)W_\hb^d(u,v)([\si])=a([\si])\int_{G_{[\si]}}(W_\hb(u,v)\circ\Phi _t)\,\mu(t), 
$$
where $\frac{\de\eta}{\de\lb}$ is the Radon-Nikodym derivative of $\eta$ with respect to the Lebesgue measure, $\alpha$ is given in Theorem \ref{main}, $G_{[\si]}$ and $a([\si])$ are given in Corollary \ref{Weins}.
\end{Theorem}
\begin{proof}
Combining $c)$ and $d)$ in theorem \ref{main} we get that
$$
\langle\Op_\hb(f)u,v\rangle=\int_{\Si}f\,W_\hb(u,v)\,\omega=\int_{J(\Si)}\int_{\Si_\lb}f_\lb([\si])\,\int_{[\si]}W_\hb(u,v)\, i^*_{[\si]}(\alpha)\,\zeta_\lb([\si])\,\de\lb.
$$
Similarly,
$$
\langle\Op^d_\hb(f)u,v\rangle=\int_{P}f_d\,W_\hb^d(u,v)\,\zeta_H=\int_{J(\Si)}\int_{\Si_\lb}f_\lb([\si])\,(\frac{\de\eta}{\de\lb}\circ J)(\si)\,W_\hb^d(u,v)([\si])\,\zeta_\lb([\si])\,\de\lb.
$$
Therefore the identity \eqref{cond} implies CQR on every constant of motion. Moreover, if we do have CQR on every constant of motion, then we can replace $f$ above by $(\varphi\circ J)f$ and we obtain 
$$
\int_{J(\Si)}\varphi(\lb)\int_{\Si_\lb}f_\lb([\si])\,\int_{[\si]}W_\hb(u,v)\, i^*_{[\si]}(\alpha)\,\zeta_\lb([\si])\,\de\lb=
$$
$$
=\int_{J(\Si)}\varphi(\lb)\int_{\Si_\lb}f_\lb([\si])\,(\frac{\de\eta}{\de\lb}\circ J)(\si)\,W_\hb^d(u,v)([\si])\,\zeta_\lb([\si])\,\de\lb.
$$
Since $\varphi$ is arbitrary, for each regular $\lb$ we have that
$$
 \int_{\Si_\lb}f_\lb([\si])\,\int_{[\si]}W_\hb(u,v)\, i^*_{[\si]}(\alpha)\,\zeta_\lb([\si])=\int_{\Si_\lb}f_\lb([\si])\,(\frac{\de\eta}{\de\lb}\circ J)(\si)\,W_\hb^d(u,v)([\si])\,\zeta_\lb([\si]).
$$
Since $f$ is arbitrary, we obtain the identity \eqref{cond}, and the last part of our result follows from Corollary \ref{Weins}.
\end{proof}

Recall that each $W_\hb^\lb$ is the Wigner transform of $\Op^\lb$, in particular, we have that 
$$
\langle u_\lb,v_\lb\rangle=\int_{\Si_\lb}W_\hb^\lb(u_\lb,v_\lb)\zeta_\lb.
$$
The latter formula together with \eqref{zeta} implies the following result.
\begin{Corollary}\label{inner}
Assume that $H_1,\cdots,H_k$ have absolutely continuous joint spectrum. A necessary condition to obtain commutation of quantization and reduction on every constant of motion is that the identity
\begin{equation}
\frac{\de\eta}{\de\lb}(\lb)\langle u_\lb,v_\lb\rangle=\int_{\hat\Si_\lb}W_\hb(u,v)\gamma_\lb, 
\end{equation}
holds for each $u,v\in\H$, where $\gamma_\lb$ is defined in remark \ref{gamma}. 
\end{Corollary}
\begin{Remark}
{\rm
In particular, a necessary condition to obtain commutation of quantization and reduction on every constant of motion is that $\int_{\hat\Si_\lb}W_\hb(u,u)\gamma_\lb\geq 0$, for each $u\in\H$. 
}
\end{Remark}

\section{Explicit examples}\label{mec}

In this section we are going to show explicit examples such that $\Si_\lb=T^*\Omega_\lb$, $\H(\lb)=L^2(\Omega_\lb)$ and $\sp(H_1,\cdots,H_k)=\overline{J(\Si)}$, where $\Omega_\lb$ is certain Riemannian manifold. 

Fortunately, the problem of quantizing $T^*\Omega$ into $L^2(\Omega)$ for an arbitrary Riemannian manifold $\Omega$ has been studied by many authors. For instance, an interesting solution can be found in Landsman's book~\cite{La} II.3 (or his article~\cite{LaYM}), which also admits a Wigner transform and coincides with Weyl quantization when $\Omega=\R^n$. For alternative solutions see~\cite{LQ,U}. Another interesting quantization scheme with a Wigner transform on symmetric spaces can be found in \cite{TE}.

We will show that theorem \ref{char} implies that, unless each $\Omega_\lb$ is flat, there is no decomposable Weyl quantization such that we have CQR on every constant of motion. The main reason why is that $\int_{[\si]}W_\hb(u,v)\,i^*_{[\si]}\alpha$ cannot define a Wigner transform on $\H(\lb)$ (proposition \ref{LW} and corollary \ref{CLW}). However, using Landsman's Wigner transform $W^\lb_\hb$ (\cite{La} II.3) on each fiber, we can construct $\Op^d_\hb$. Since  $W^\lb_\hb$ is very similar to $\int_{[\si]}W_\hb(u,v)\,i^*_{[\si]}\alpha$, we expect to obtain SCQR on every constant of motion. Moreover, there are some generic examples on which we still have CQR. For instance, we show that $\Op(\varphi \circ J)=\Op^d(\varphi \circ J)=\varphi(H_1,\cdots, H_k)$. We also prove that quantization and reduction commute on each $f=J_Y$, whenever $Y$ is a vector field tangent to each $\Omega_\lb$ and $J_Y(x,\xi)=\langle Y(x),\xi\rangle$. 

In what follows, we will barely need to consider the role played by Planck's constant $\hb$, so we will include it in the notation only if it necessary or it seems important to do so.  

The original phase space in the examples of this section is  $\Si=T^*\R^n$ with its canonical symplectic form, thus we can quantize $\Si$ using the canonical Weyl quantization and consider the corresponding problem of commutation of quantization and reduction. However, we shall leave this problem for the future.

Although the classical Hamiltonians $h_1,\cdots,h_k$ and the quantum Hamiltonians $H_1,\cdots, H_k$ that we consider in this section are well known, recall that the focus of our analysis is not to study the Hamiltonians but their constants of motion in relation to quantization. Also notice that the algebra of constant of motion can be quite large in that cases; for instance, if $k=1$ and $h$ is the free Hamiltonian, then the algebra of constant of motion has a copy of $C^\infty(\R)\otimes C^\infty(\mathfrak{so}(n)^*)$ as a subalgebra (the first term are compositions of $h$ with functions on $\R$ and the second term are functions of angular momenta).

\subsection{The classical side.}
Let us begin giving a simple tool to compute $\Si_\lb$.
\begin{proposition}\label{cequi}
Let $h'_1,\cdots,h'_k\in C^\infty(\Si)$ such that $\{h'_i,h'_j\}=0$. Also, let $S$ be a symplectomorphism on $\Si$ and define $h_j=h'_j\circ S$. Denote by $\Si_\lb$ and $\Si'_\lb$ the reduced spaces corresponding to $h_1,\cdots, h_k$ and $h'_1,\cdots, h'_k$ respectively. The map $ S_\lb:\Si_\lb\to\Si'_\lb$ given by $S_\lb[\si]=[S\si]$ is a symplectomorphism. The map $f\to f\circ S$ defines an isomorphism between the corresponding Poisson algebras of constants of motion.
\end{proposition}
\begin{proof}
It is easy to show that $S(\hat\Si_\lb)=\hat\Si_\lb'$ and $\Phi_t=S^{-1}\circ\Phi'_t\circ S$. Thus, $S_\lb$ is a diffeomorphism. Let $\Theta_\lb$ and $\Theta'_\lb$ be the reduced symplectic forms on $\Si_\lb$ and $\Si_\lb'$ respectively. We must show that $ S_\lb^*(\Theta'_\lb)=\Theta_\lb$. We know that $\Theta_\lb$ is characterized by the condition $\pi_\lb^*(\Theta_\lb)=i^*_\lb(\Theta)$, where $\pi_\lb:\hat\Si_\lb\to\Si_\lb$ is the quotient map, $i_\lb:\hat\Si_\lb\to\Si$ is the inclusion map and $\Theta$ is the symplectic form on $\Si$. Notice that $\pi'_\lb\circ S= S_\lb\circ\pi_\lb$ and $i'_\lb\circ S=S\circ i_\lb$. Then
$$
\pi_\lb^*\circ S_\lb^*(\Theta'_\lb)=(S_\lb\circ\pi_\lb)^*(\Theta'_\lb)=S^*\circ(\pi'_\lb)^*(\Theta'_\lb)=S^*\circ(i'_\lb)^*(\Theta)=i_\lb^*\circ S^*(\Theta)=(i_\lb)^*\Theta.
$$ 
The last part of this proposition is straightforward.     
\end{proof}

In what follows we are going to consider the case $\Si=T^*\R^n$, mainly because we can construct examples using the metaplectic representation. However, theorem \ref{sym} and theorem \ref{DI} can be easily generalized to the case $\Si=T^*\Omega$. 

Let $h_j(x,\xi)=\phi_j(x)$, where $\phi_j\in C^\infty(\R^n)$. Also let $\tilde J=(\phi_1,\cdots,\phi_k):\R^n\to\R^k$ and let
$$
\Omega_{\lb}:=\tilde J^{-1}(\lb),
$$
for each $\lb\in \tilde J(\R^{n})$. Thus
$$
\hat{\Si}_{\lb}=J^{-1}(\lb)=\Omega_{\lb}\times\R^n
$$
and
$$
\Phi_{t}(x,\xi)=
(x,\xi+t_1\nabla\phi_1(x)+\cdots+t_k\nabla\phi_k(x))\,\,\,\forall t\in\R^k\, ,(x,\xi)\in\R^{2n}.
$$
If $\lb$ is regular, by the implicit function theorem, $\Omega_{\lb}$
is a $n-k$ submanifold of $\R^n$ and $\nabla\phi_j(x)$ is normal at each $x\in\Omega_{\lb}$,
i.e. $\nabla\phi_j(x)\in [(i_\lb)_*(T_x\Omega_{\lb})]^\bot$, where
$i_{\lb}:\Omega_{\lb}\to\R^n$ is the inclusion and $T_x\R^n$ is identified
with $\R^n$ in the usual way.

Let $g$ be the metric on $\Omega_{\lb}$ induced from the standard metric on $\R^n$, i.e. $g_x:T_x\Omega_{\lb}\times T_x\Omega_{\lb}\to\R$ is given by
$$
g_x(v,w):=\langle(i_\lb)_* v,(i_\lb)_* w\rangle.
$$

Also, let $\tilde{g}_x:T_x\Omega_{\lb}\to T_x^*\Omega_{\lb}$ be the natural
isomorphism coming from $g_x$. We identify $T^*_x\Omega_{\lb}$ with the subspace $<\nabla\phi_1(x),\cdots,\nabla\phi_k(x)>^\bot$ using the map $i^\lb_*\circ\tilde{g}_x^{-1}$. Let us denote by $q_x$ be the orthogonal projection of $\R^n$ onto $<\nabla\phi_1(x),\cdots,\nabla\phi_k(x)>^\bot$.

\begin{Theorem}\label{sym}
Let $\phi_j\in C^\infty(\R^n)$ and $h_j(x,\xi)=\phi_j(x)$. If $\lb$ is regular, then:
\begin{equation}\label{symp}
\Si_{\lb}\ni [(x,\xi)]\to (x,q_x(\xi))\in
T^*\Omega_{\lb}
\end{equation}
is a symplectomorphism, where $[(x,\xi)]$ denote the orbit of $(x,\xi)$ by $\Phi$. Here $T^*\Omega_{\lb}$
is endowed with the standard symplectic structure on a cotangent bundle and $T_x^*\Omega_{\lb}$
is identified as above with the $(n-k)$-dimensional subspace $<\nabla\phi_1(x),\cdots,\nabla\phi_k(x)>^\bot$.
\end{Theorem}

\begin{proof}
It is clear that the map (\ref{symp}) is well defined and it is a diffeomorphism. The closed $2$-form
$\Theta_{\lb}$ in $\Si_{\lb}$
is determined by the equation
$$
i_{\lb}^*\Theta=\pi_{\lb}^*\Theta_{\lb},
$$
where $\Theta=\sum_k \de \xi_k\wedge\de x_k$ is the standard closed $2$-form on $\R^{2n}$ and
$\pi_{\lb}:\hat{\Si}_{\lb}\to\Si_{\lb}$ is the quotient map onto the orbit
space (see \cite{Me} or \cite{MW} or \cite{M}). Thus, we only need to check that the standard closed $2$-form $\Theta_{\lb}$ in
$T^*\Omega_{\lb}$ satisfies that equation, after replacing $\pi_{\lb}$ by
$\tilde\pi_{\lb}:\hat\Si_{\lb}\to T^*\Omega_{\lb}$, where
$\tilde\pi_{\lb}((x,\xi))=(x,q_x(\xi))$.

Fix $(x_0,\xi_0)\in\hat{\Si}_{\lb}$, it is enough to prove that there is a basis $\{v_j\}_{j=1}^{2n-k}$ for
$T_{(x_0,\xi_0)}\hat{\Si}_{\lb}$ such that
\begin{equation}\label{redprop}
\Theta_{(x_0,\xi_0)}\left((i_{\lb})_*v_l,(i_{\lb})_*v_m\right)=
(\Theta_{\lb})_{\tilde\pi_{\lb}(x_0,\xi_0)}\left((\tilde\pi_{\lb})_*v_l,
(\tilde\pi_{\lb})_*v_m\right).
\end{equation}
Let us fix a chart $(z_1,\cdots,z_{n-k};U)$ at $x_0\in\Omega_\lb$ and let
$\{\frac{\partial}{\partial z_j}|_{x_0}\}_{j=1}^{n-k}$ the corresponding basis of $T_{x_0}\Omega_\lb$ (recall that $\Omega_\lb$ is a $(n-k)$-submanifold of $\R^n$). Also let $e_j=\frac{\partial}{\partial x_j}|_{(x_0,\xi_0)}, f_j=\frac{\partial}{\partial \xi_j}|_{(x_0,\xi_0)}$ the canonical basis of $T_{(x_0,\xi_0)}\R^{2n}$ and $e_j^*=\de x_j|_{(x_0,\xi_0)},f_j^*=d\xi_j|_{(x_0,\xi_0)}$ the corresponding dual basis of $T^*_{(x_0,\xi_0)}\R^{2n}$. Since $\hat\Si_\lb=\Omega_\lb\times\R^n$, we can consider the basis of $T_{(x_0,\xi_0)}\hat{\Si}_{\lb}$ given by
$$
v_j=
\left\{
  \begin{array}{ll}
\frac{\partial}{\partial z_j}|_{(x_0,\xi_0)}  &\, 1\leq j\leq n-k\\
 \frac{\partial}{\partial \xi_{(j+k-n)}}|_{(x_0,\xi_0)} &\,n-k< j\leq 2n-k.
  \end{array}
\right.
$$

Let us calculate the left hand side of \eqref{redprop}. Let $w_j=(i_\lb)_*(\frac{\partial}{\partial_{ z_j}}|_{x_0})$, then clearly
$$
(i_\lb)_* v_j=
\left\{
  \begin{array}{ll}
\sum_m \langle w_j,e_m\rangle e_m &\,j\leq n-k \\
f_{j+k-n} &\,j>n-k.
\end{array}
\right.
$$
By definition
$$
\Theta_{(x_0,\xi_0)}\left((i_\lb)_* v_j,(i_\lb)_* v_l\right)=\frac{1}{2}\sum_m det\left(
                                                 \begin{array}{cc}
                                            f^*_m((i_\lb)_*v_j) & f^*_m((i_\lb)_*v_l)\\
                                            e^*_m((i_\lb)_*v_j) & e^*_m((i_\lb)_*v_l) \\
                                                 \end{array}
                                               \right).
$$
So, if $j>n-k$ and $l\leq n-k$, we get
$$
\frac{1}{2}\sum_m det\left(
\begin{array}{cc}
f^*_m(f_{j+k-n}) &
f^*_m(\sum_{r=1}^n\langle w_l,e_r\rangle e_r)\\
e^*_m(f_{j+k-n}) &
e^*_m(\sum_{r=1}^n\langle w_l,e_r\rangle e_r) \\
                                                 \end{array}
                                               \right)=
$$
$$
=1/2\langle w_l,e_{j+k-n}\rangle .
$$
The other cases can be obtained in the same way, and we have that
$$
\Theta_{(x_0,\xi_0)}\left((i_\lb)_* v_j,(i_\lb)_*v_l\right)=
\left\{
  \begin{array}{ll}
    -1/2 \langle w_j,e_{l+k-n}\rangle &\,j\leq n-k,\, l> n-k \\
1/2 \langle w_l,e_{j+k-n}\rangle &\, j>n-k,\, l\leq n-k\\
0  &\, \text{otherwise}.
  \end{array}
\right.
$$
For the right hand side of equation \eqref{redprop}, let us denote by $(z,p)$ the elements of $T^*\Omega_\lb$ and by $p_j$ the coordinates on the cotangent part corresponding to the dual basis $\de z_j|_z$. Recall that for any vector field $X$ on $T^*\Omega_\lb$, we have that
$$
X=\sum_{m=1}^{n-k}a_m\frac{\partial}{\partial z_m}|_{\tilde\pi_\lb(x_0,\xi_0)}+b_m\frac{\partial}
{\partial p_m}|_{\tilde\pi_\lb(x_0,\xi_0)}
$$
where $a_m=X(z_m)$ and $b_m=X(p_m)$. Thus, when $X=(\tilde\pi_\lb)_* v_j$, we have that $a_m=v_j(z_m\circ\tilde\pi_\lb)$ and $b_m=v_j(p_m\circ\tilde\pi_\lb)$. Clearly 
$z_m\circ\tilde\pi_\lb(x,\xi)=z_m(x)$. Therefore, if $j\leq n-k$, we have that $(\tilde\pi_\lb)_* v_j =\frac{\partial}{\partial z_j}|_{\tilde\pi^\lb((x_0,\xi_0))}$.

We still need to compute $p_m\circ\tilde\pi_\lb$.
Let $y_l:=(i_\lb)_*\circ\tilde{g}_{x_0}^{-1}(\de z_l|_{x_0})$. Clearly $\{y_l\}_{l=1}^{n-k}$ is a basis of $i^\lb_*(T_{x_0}\Omega_\lb)$, but it is not necessarily orthonormal. However, we have that
$$
\langle y_l,w_m\rangle=g_{x_0}\left(\tilde{g}_{x_0}^{-1}(\de z_l|_{x_0}),\frac{\partial}{\partial z_m}|_{x_0}\right)=\de z_l|_{x_0}
(\frac{\partial}{\partial z_m}|_{x_0})=\delta_l^m.
$$
Notice that $\langle q_{x_0}(\xi),w_j\rangle=\langle \xi,w_j\rangle$. Thus, $\langle q_{x_0}(\xi),w_j\rangle=\left\langle\sum_l\langle \xi,w_l\rangle y_l,w_j\right\rangle$ for every $\xi\in\R^n$, and the non-degeneracy of the inner product implies that
$$
q_{x_0}(\xi)=\sum_l\langle \xi,w_l\rangle y_l.
$$
Therefore
$$
p_m\circ\tilde\pi_\lb(x,\xi)=p_m\left[\tilde{g}_x\circ(i_\lb)_*^{-1}\left(\sum_{l=1}^{n-k}\langle\xi,w_l\rangle y_l\right)\right]=
p_m\left(\sum_{l=1}^{n-k}\langle \xi,w_l\rangle\de z_l\right)=\langle \xi,w_m\rangle.
$$
So, for $j>n-k$, $v_j(p_m\circ\tilde\pi_\lb)=\langle e_{j+k-n},w_m\rangle$. Thus
$$
(\tilde\pi_\lb)_* v_j=
\left\{
  \begin{array}{ll}
\frac{\partial}{\partial z_j}|_{\tilde\pi^\lb((x_0,\xi_0))}  &\, j\leq n-k\\
\sum_m\langle e_{j+k-n},w_m\rangle \frac{\partial}{\partial p_m}|_{\tilde\pi_\lb(x_0,\xi_0)} &\,j> n-k.
  \end{array}
\right.
$$
Finally, computing in the same way that we did for the left hand side of \eqref{redprop}, we obtain that
$$
(\Theta_\lb)_{\tilde\pi_\lb(x_0,\xi_0)}\left((\tilde\pi_\lb)_* v_j,(\tilde\pi_\lb)_* v_l\right)=
\left\{
  \begin{array}{ll}
    -1/2 \langle w_j,e_{l+k-n}\rangle &;\,j\leq n-k,\, l> n-k \\
1/2 \langle w_l,e_{j+k-n}\rangle &;\, j> n-k,\, l\leq n-k\\
0  &;\, \text{otherwise}.
  \end{array}
\right.
$$
\end{proof}

\begin{Corollary}
Let $h_j=h_j'\circ S$, for $j=1,\cdots, k$, where $S$ is a simplectomorphism and $h'_j(x,\xi)=\phi_j(x)$. Then $\Si_\lb\cong T^*\Omega_\lb$. In particular, if $h_j(x,\xi)=\phi_j(\xi)$, then  $\Si_\lb\cong T^*\Omega_\lb$.  
\end{Corollary}
\begin{proof}
The main part is a direct consequence of proposition \ref{cequi}. For the second part it is enough to take $h'_j(x,\xi)=\phi_j(x)$ and $S=\mathfrak F$ is the symplectic matrix.
\end{proof}


The case $h_j(x,\xi)=\phi_j(\xi)$ is important for physical applications. The proof of our previous theorem can be adapted trivially to this case, but we wanted to emphasize the role played by the symplectic matrix. An important example of an only momentum depending Hamiltonian is the free Hamiltonian $h(x,\xi)=|\xi|^2/2$. Clearly, in the latter case, $\nabla\phi(\xi)=\xi$ and $\Omega_\lb=\mathbb S^{n-1}_{\sqrt{2\lb}}$ is the $n-1$-dimensional sphere of radius $\sqrt{2\lb}$ for each $\lb>0$. 

\subsection{The quantum side.}
The quantum version of proposition \ref{cequi} is also straightforward, it is a direct consequence of the uniqueness of the functional calculus and the uniqueness of the simultaneous diagonalization.
\begin{proposition}\label{qequi}
Let $H'_1,\cdots,H'_k$ a pairwise strongly commuting family of selfadjoint operators on a Hilbert space $\H$ and $T':\H\to\int_{\sp(H_1,\cdots,H_k)}^\oplus \H(\lb)\de\eta(\lb)$ its simultaneous diagonalization. Also, let $\mathcal S$ be an unitary operator on $\H$ and define $H_j=\mathcal S^*H'_j \mathcal S$. Then $H_1,\cdots, H_k$ is a pairwise strongly commuting family of selfadjoint operators and $T:=T'\mathcal S$ is its simultaneous diagonalization.    
\end{proposition}

The canonical quantization of $\R^{2n}$ into $L^2(\R^n)$ is the Weyl quantization, which we denote by $\Op_\hb$. So the quantum counterpart of $h\in C^\infty(\R^{2n})\cap S'(\R^{2n})$ is $H=\Op^\hb(h)$, where $S'(\R^{2n})$ is the space of tempered distributions. 

Notice that if $S$ is a symplectomorphism then $\Op_\hb(h)$ and $\Op^\hb(h\circ S)$ are not necessarily unitarily equivalent (Egorov's theorem states that this is true when $\hb\to 0$ though). However, when $S$ is also linear the unitary equivalence holds and the required unitary operator is given by the so called metaplectic representation $\mu$. More precisely, we know that for any $f\in S'(\R^{2n})$, and any linear symplectomorphism $S$, we have that 
$$
\Op_\hb(f\circ S) = \mu(S^*)\Op_\hb(f)\mu(S^*)^{-1}.
$$
For instance, see theorem 2.15 in \cite{Fol} for details.
\begin{Corollary}\label{met}
If $T'$ is the simultaneous diagonalization of a strongly commuting family of selfadjoint operators of the form $\Op_\hb(h'_1),\cdots,\Op^\hb(h'_k)$ and $S$ is a linear symplectomorphism, then $T:=T'\mu(S^*)^{-1}$ is the simultaneous diagonalization of  $\Op^\hb(h_1),\cdots,\Op^\hb(h_k)$, where $h_j=h_j'\circ S$.
\end{Corollary}

If $h_j(x,\xi)=\phi_j(x)$ as before, its natural quantum counterpart is $H_j:=\phi_j(Q)$ the operator multiplication by $\phi_j$. Clearly,
$\{H_j\}_{j=1}^k$ is a pairwise commuting family of selfadjoint operators. The notation is meant to emphasize that $\phi(Q)$ is the operator given by the simultaneous functional calculus of the family of strongly commuting operators $Q_1,\cdots, Q_n$, where $Q_j u(x)=x_j u(x)$. Notice that $\Op_\hb(h_j)=H_j$. We can present the simultaneous diagonalization for $H_j:=\phi_j(Q)$ in two equivalent ways: the first one seems to be more natural, but the second one is more practical for our purposes and it is based on Corollary \ref{inner}. We shall described both of them.

The simultaneous functional calculus when $H_j:=\phi_j(Q)$ is given by $\varphi(H_1,\cdots,H_k)u(x)=[\varphi\circ\tilde J](x)u(x)$, where $\varphi$ is a Borel function on
$\R^k$ and $u\in L^2(\R^n)$; in particular $sp(H_1,\cdots, H_k)=\overline{\tilde J(\R^n)}=\overline{J(\R^{2n})}$. The latter fact suggests that, in order to find the simultaneous diagonalization of the family $\{H_j\}_{j=1}^k$, we should consider the measurable  field of Hilbert spaces $\{J(\R^{2n})\ni\lb\to L^2(\Omega_\lb,\eta_\lb)\}$ and we should look for a unitary operator of the form $Tu(\lb)=(\rho u)|_{\Omega_\lb}$, where $\eta_\lb$ is the Borel measure defined on $\Omega_\lb$ given by its induced Riemannian structure and $\rho$ is some suitable function on $\R^n$ placed there only to ensure that $T$ becomes unitary.

Notice that $J$ is regular at $(x,\xi)\in\R^{2n}$ iff $\tilde J$ is regular at $x$, i.e. iff the Jacobian
$D\tilde J(x):\R^n\to\R^k$ has rank $k$, or equivalently, 
$\wedge^k [D\tilde J(x)]:\wedge^k\R^n\to\wedge^k\R^k\cong\R$ is not the trivial functional.

Let
\begin{equation*}
\mathcal I=\left\{\lb\in \tilde J(\R^n)/\exists x\in \Omega_\lb,\, \text{such that}\,
\wedge^k D\tilde J(x)=0\right\}.
\end{equation*}
The Morse-Sard Theorem asserts that $\mathcal L^k(\mathcal I)=0$, where $\mathcal L^k$ is the Lebesgue measure.

The key ingredient to find the diagonalization is the so-called coarea formula, which allows us to compute explicitly the volume forms $\gamma_p$ in lemma \ref{mainL}: for any non-negative measurable function $f$ on $\R^n$, the function
$\R^k\ni\lb\to \int_{\Omega_\lb}f(z)\de \eta_\lb(z)$ is measurable, and we have that
 \begin{equation}\label{gcoarea}
\int_{\R^n}f(x)||\wedge^k D\tilde J(x)||\de x=\int_{\R^k}\int_{\Omega_\lb}f(z)\de \eta_\lb(z)\de \lb,
\end{equation}
where the norm involved is the usual operator norm on linear maps. The latter identity also holds for any measurable function $f$ on $\R^n$ if one of the two sides of \eqref{gcoarea} is well defined.

Coarea formula can be found in \cite{Si} 10.6 or it can also be
easily deduced from Theorem 3.2.11 in \cite{Fe}, and it is easy to show that it also holds true if we replace $\R^n$ by an arbitrary Riemannian manifold. 

We now look for the required measure $\eta$ on $\overline{J(\R^{2n})}$. In the literature, $\eta$ is called the scalar spectral measure of $H_1,\cdots, H_k$ (or a basic measure in the framework of operator algebras \cite{D}). It is determined by the condition 
$\eta(A)=0$ iff $\eta_u(A)=0$ for every $u\in L^2(\R^n)$, where $\eta_u$ is the measure defined
by $\eta_u(A)=<\chi_A(H_1,\cdots,H_k)u,u>$ and $A\subset\R^k$ is any Borel set. Equivalently, $\eta_u\ll\eta$ for each $u\in\H$ and, if $\mu$ is any other measure such that $\eta_u\ll\mu$ for each $u\in\H$, then $\eta\ll\mu$. So, the Radon-Nikodym's theorem implies that $\eta$ is unique up to equivalence. 

\begin{Lemma}\label{mea}
Assume that $\mathcal L^n\{x\in\R^n/\wedge^k D\tilde J(x)= 0\}=0$, where $\mathcal L^n$ is the Lebesgue measure. The scalar spectral measure is given by $\eta(A)=\mathcal{L}^k(\tilde J(\R^n)\cap A)$. 
\end{Lemma}

The reader might wonder why $\eta$ is not only the restriction of $\mathcal L^k$ to $\tilde J(\R^n)$. The latter will be the case if
$\mathcal L^k(\overline{\tilde J(\R^n)}\backslash \tilde J(\R^n))=0$. However, such equality is true for $k=1$, but it fails for $k>1$ in general. When $k=2$, an example can be constructed as follows: Let $C$ be a Jordan curve in the plane with positive area, thus its interior region is open and simply connected and its boundary is $C$. The Weierstrass uniformization Theorem implies that the interior of $C$ is biholomorphic with the open disc. Therefore, it is enough to take
$\tilde J$ to be the composition of such holomorphic function with a smooth map from $\R^n$ onto the open disc.

\begin{proof}
Note that, for every Borel set $A\subset\R^k$,
$$
<\chi_A(H_1,\cdots,H_k)u,u>=\int_{\R^n}\chi_A(\tilde J(x))|u(x)|^2\de x=
$$
$$
\int_{\R^k}\int_{\Omega_\lb}\chi_A(\tilde J(z))||\wedge^k D\tilde J(z)||^{-1}|u(z)|^2\de\eta_\lb(z)\de\lb=
\int_{A\cap \tilde J(\R^n)}\int_{\Omega_\lb}||\wedge^k D\tilde J(z)||^{-1}|u(z)|^2\de\eta_\lb(z)\de\lb,
$$
then $\eta_u(A)=0$ for every $u\in L^2(\R^n)$ iff $\mathcal L^k(\tilde J(\R^n)\cap A)=0$.
\end{proof}

The following result, even if it follows almost directly from coarea formula, it is not stated elsewhere to the state of knowledge of the authors (except for the well known case $k=1$ and $\phi(x)=\frac{x^2}{2}$, or in general for $\phi(x)=h_0(||x||)$, lemma 3.6 in \cite{T}).
Essentially, the following theorem states that we should take $\rho(x)=||\wedge^k D\tilde J(x)||^{-1/2}$.

\begin{Theorem}\label{DI}

Let $\mathcal M=\{x\in\R^n:\wedge^k D\tilde J(x)\neq 0\}$ and define $\rho:\mathcal M\to\R$ by
$\rho(x)=||\wedge^k D\tilde J(x)||^{-1/2}$. Then the map
$$
T_0: L^2(\mathcal M)\to \int^\bigoplus_{\overline{\tilde J(\R^n)}} L^2(\Omega_\lb,\de\eta_\lb)\de\eta(\lb),
$$
given by
$$
[T_0 u(\lb)](z):=\left\{
  \begin{array}{ll}
    \rho(z)u(z)\,\,\, &\text{if } \lb\in \tilde J(\R^n)\backslash \mathcal I\\
0  & \text{otherwise}.
  \end{array}
\right.
$$
is unitary. If $\mathcal L^n(\mathcal M^c)=0$ then $T_0$ simultaneously diagonalize the family $\{H_j:=\phi_j(Q)\}_{j=1}^k$.
\end{Theorem}

\begin{proof}
It follows from Sard's lemma and coarea formula that $T_0$ is well defined and unitary. Moreover, for each function $\varphi$
Borel and bounded on $\R^k$, $z\in\Omega_\lb$ and $u\in L^2(\R^n)$, we have that
$$
[T_0 \varphi(H_1,\cdots,H_k)u(\lb)](z)=\rho(z)[\varphi(H_1,\cdots,H_k)u](z)=
\rho(z)f(\tilde J(z))u(z)= \varphi(\lb)[T_0u(\lb)](z).
$$
\end{proof}
\begin{Remark}\label{posdep}

{\rm
This theorem would also follows under milder conditions over the family $\{\phi_j\}$, for instance it is enough to
assume that $\tilde J$ is Lipschitz on any bounded set.}
\end{Remark}

If $h_j(x,\xi)=\phi_j(\xi)$, then its natural quantum counterpart is $H_j:=\phi_j(P):=\mathcal F^{-1}\phi_j(Q)\mathcal F$, where $\mathcal F$ is the Fourier transform. The notation is meant to emphasize that $\phi(P)$ is the operator given by the simultaneous functional calculus of the family of strongly commuting operators $P_1,\cdots, P_n$, where $P_j u(x)=-i\hb\partial_j u(x)$. Notice that $\Op^\hb(h_j)=H_j$. Proposition \ref{qequi} implies that $T_0\mathcal F$ is the diagonalization of $H_1,\cdots, H_k$ in this case. 

We can combine Theorem \ref{DI} with the metaplectic representation to compute diagonalizations.
\begin{Corollary}
Let $h'_j(x,\xi)=\phi_j(x)$ and $S$ be a linear symplectomorphism. Define $h_j=h'_j\circ S$ and assume that $\mathcal L(\mathcal M^c)=0$. Then $\Op^\hb(h_1),\cdots,\Op^\hb(h_k)$ is a strongly commuting family of selfadjoint operators and $T=T_0\mu(S^*)^{-1}$ is its simultaneous diagonalization.   
\end{Corollary}
For example, we recover the example $H_j=\phi_j(P)$ described above after noticing that $\mu(\mathfrak F)=\mathcal F$.
 
As we mention above, there is another way to present theorem \ref{DI} much more useful for our purposes. Motivated by corollary \ref{inner}, consider the family of sesquilinear forms defined on $L^2(\R^n)$ given by 
$$
(u,v)_\lb:=\int_{\hat\Si_\lb}W_{u,v}\gamma_\lb.
$$ 
We will show that $(\cdot,\cdot)_\lb$ is nonnegative; but it might fail to be an inner product. However, we can still construct Hilbert spaces canonically: consider the set $K_\lb=\{u\in L^2(\R^n):(u,u)_\lb=0\}$, the Cauchy-Schwartz inequality implies that $K_\lb$ is a subspace and $(\cdot,\cdot)_\lb$ defines an inner product on the quotient space $L^2(\R^n)/K_\lb$. We denote by $\H_\lb$ the Hilbert space obtained by completing $L^2(\R^n)/K_\lb$ and we denote by $p_\lb:L^2(\R^n)\to \H_\lb$ the canonical projection (with dense range). We shall prove that we recover the diagonalization of $H_1,\cdots, H_k$ taking $T_0=\int^\oplus p_\lb\eta(\lb)$, but in order to do it we need to compute first the volume form $\gamma_\lb$, and for the next subsection we also compute the volume form  $\alpha_{[\si]}$.

\begin{Lemma}
The volume form $\gamma_\lb$ on $\hat\Si_\lb=\Omega_\lb\times\R^n_\xi$ coincides with $\rho^2\eta_\lb\wedge\de\xi$. The volume form $\alpha_{[x,\xi]}$ on the orbit $[x,\xi]$ coincides with  $\rho^2\mu_{(x,\xi)}$, where $\mu_{(x,\xi)}$ is the volume form given by the Riemannian structure induced on $[x,\xi]=\xi+\langle\nabla\phi_j(x)/j=1,\cdots k\rangle$. 
\end{Lemma}
\begin{proof}
First notice that there is $\varphi\in C^\infty(\hat\Si_\lb)$ such that $\rho^2\eta_\lb\wedge\de\xi=\varphi\gamma_\lb$. Combining coarea formula and equation \eqref{gamma}, we have that
$$
\int_{\R^{2n}}f\de\xi\de x=\int_{\tilde J(\R^n)}\left(\int_{\tilde\Si_\lb} f \rho^2\eta_\lb\wedge\de\xi\right)\de\lb=
$$
$$
=\int_{\tilde J(\R^n)}\left(\int_{\tilde\Si_\lb} f \varphi\gamma_\lb\right)\de\lb=\int_{\R^{2n}}f\varphi\de\xi\de x.
$$
Therefore $\varphi=1$. For the second part, fix $x\in\Omega_\lb$ and let $\pi_x:T_x\R^n\to T^*_x\Omega_\lb$ the orthogonal projection. Coarea formula implies that, for each $ f\in L^1(T^*_x\R^n)$, we have that
$$
\int_{T_x^*\R^n}f=\int_{T^*_x\Omega_\lb} \left[\int_{\pi_x^{-1}(\xi)}f\mu_{(x,\xi)}\right]\de(\xi)
$$
where $\mu_{(x,\xi)}$ is the Riemannian volume form on  $\pi_x^{-1}(\xi)=[x,\xi]$. Using equation \eqref{zeta}, we have that
$$
\int_{\Omega_\lb}\left(\int_{T^*_x\Omega_\lb} \left[\int_{\pi_x^{-1}(\xi)}f\alpha_{[x,\xi]}\right]\de\xi\right)\eta_\lb(x)=\int_{\hat\Si_\lb}f\gamma_\lb=\int_{\Omega_\lb}\left(\int_{T^*_x\Omega_\lb} \left[\int_{\pi_x^{-1}(\xi)}f\rho^2\mu_{(x,\xi)}\right]\de\xi\right)\eta_\lb(x). 
$$
Therefore, using the same argument we used to compute $\gamma_\lb$, we obtain that $\alpha_{[x,\xi]}=\rho^2\mu_{[x,\xi]}$.
\end{proof}
\begin{proposition}
Let $h_j(x,\xi)=\phi_j(x)$ with $j=1,\cdots,k<n$, and assume that $\mathcal L^n\{x\in\R^n:\wedge^k D\tilde J(x)= 0\}=0$.
If $\lb\in \tilde J(\R^n)/\mathcal I$, then the sesquilinear form $(\cdot,\cdot)_\lb$ is nonnegative and the Hilbert space $\H_\lb$ coincides with $L^2(\Omega_\lb,\rho^2\eta_\lb)$. Moreover, the map $T_0:L^2(\R^n)\to \int^\bigoplus_{\overline{\tilde J(\R^n)}}\H_\lb\de\eta(\lb)$, given by $T_0u(\lb)=p_\lb(u)$, is the simultaneous diagonalization of the family of strongly commuting selfadjoint operators $H_j:=\phi_j(Q)$, with $j=1,\cdots,k$.
\end{proposition}
     
 \begin{proof}
It is well know that $\int_{\R^n} W_{u,v}(x,\xi)\de\xi=u(x)\overline{v(x)}$ (proposition 1.96 in \cite{Fol}). Thus,
$$
(u,v)_\lb=\int_{\Omega_\lb}u(x)\overline{v(x)}\rho^2(x)\eta_\lb(x).
$$
In particular, $(\cdot,\cdot)_\lb$ is nonnegative and $(u,u)_\lb=0$ iff $u|_{\Omega_\lb}=0$.
\end{proof}   
\subsection{The field of Wigner transforms $W^\lb$.}
We have already compute the symplectic manifolds $\Si_\lb$ and the Hilbert spaces $\H(\lb)$, so in order to construct the decomposable Weyl quantization, we must define a field of Wigner transforms. In view of theorem \ref{char}, the natural candidate would be
$$
\tilde W^\lb_{u,v}([\si]):=\int_{[\si]}W_{u,v}\alpha_{[\si]}.
$$
Notice that, by construction we have that
$$
\int_{\Si_\lb}\tilde W^\lb_{u,v}\zeta_\lb=(u,v)_\lb.
$$
However, $\tilde W^\lb$ defines a sesquilinear form on $L^2(\R^n)$, but not necessarily in $\H(\lb)$. In order to do so, since $\H(\lb)$ is a quotient space (up to completion), we would need to show that $\tilde W^\lb_{u,u}=0$ whenever $(u,u)_\lb=0$. However, we will show that such condition does not hold unless $\Omega_\lb$ is flat.

Recall that the Wigner transform $W$ of Weyl quantization is given by
$$
W_{u,v}(x,\xi)=\int_{\R^n}e^{-i\xi\cdot p}u(x+p/2)\overline{v}(x-p/2)\de p.
$$

\begin{proposition}\label{LW}
If $x\in\Omega_\lb$ and $\xi\in T_x^*\Omega_\lb\cong\nabla\phi^\perp$, then
$$
\tilde W^\lb_{u,v}(x,\xi)=\rho^2(x)\int_{T_x\Omega_\lb}e^{-i\xi\cdot p}u(x+p/2)\overline{v}(x-p/2)\de p
$$
\end{proposition} 
\begin{proof}

$$
\int_{[x,\xi]}W_{u,v}\alpha_{[x,\xi]}=\rho^2(x)\int_{\pi_x^{-1}(\xi)}\left(\int_{T_x\Omega_\lb} \left[\int_{\pi_x^{-1}(p)}e^{i\zeta\cdot q}u(x+q/2)\overline{v}(x-q/2)\mu_{(x,p)}(q)\right]\de p\right)\mu_{(x,\xi)}(\zeta)=
$$
$$
=\rho^2(x)\int_{T_x\Omega_\lb}\left( \int_{\pi_x^{-1}(\xi)}\left[\int_{\pi_x^{-1}(p)}e^{i\zeta\cdot q}u(x+q/2)\overline{v}(x-q/2)\mu_{(x,p)}(q)\right]\mu_{(x,\xi)}(\zeta)\right)\de p=
$$
$$
=\rho^2(x)\int_{T_x\Omega_\lb}\left( \int_{\pi_x^{-1}(0)}\left[\int_{\pi_x^{-1}(0)}e^{i(\xi+\zeta)\cdot (p+q)}u(x+p/2+q/2)\overline{v}(x-p/2-q/2)\mu_{(x,0)}(q)\right]\mu_{x,0}(\zeta)\right)\de p=
$$
$$
=\rho^2(x)\int_{T_x\Omega_\lb}e^{i\xi\cdot p}u(x+p/2)\overline{v}(x-p/2)\de p
$$
\end{proof}

In particular, $\tilde W_\lb(u,u)$ depends on the restriction of $u$ to each $T_x\Omega_\lb$, but not on the restriction of $u$ to $\Omega_\lb$. In other words,  $u_\lb=u|_{\Omega_\lb}=0$ does not implies $\tilde W_\lb(u,u)=0$, unless $\Omega_\lb$ is flat. Notice that $\Omega_\lb$ is flat exactly when each $\phi_j$ is linear and in such case $\tilde W_\lb$ is the canonical Wigner transform. Using the metaplectic representation, we obtain the following consequence.
\begin{Corollary}\label{CLW}
If $h_1,\cdots h_k$ is a family of linear functions on $\R^{2n}$ such that $\{h_i,h_j\}=0$ for any $i,j\leq k<n$, then quantization and reduction commutes on every constant of motion.
\end{Corollary}
\begin{proof}
Lemma 2.17 in \cite{Fol} implies that there is a linear symplectomorphism $S$ such that $h_j=q_j\circ S$, where $q_j(x)=x_j$. So, our result follows from the previous proposition, corollary \ref{met} and theorem \ref{char}.
\end{proof} 

Even if $\tilde W_\lb$ does not defines a Wigner transform, we can perturb it to do so. Indeed we can define 
\begin{equation}\label{wign}
W^\lb_{u_\lb,v_\lb}([x,\xi]):=\rho^2(x)\int_{T_x\Omega_\lb}e^{-i\xi\cdot p}u(\text{exp}_x (p/2))\overline{v}(\text{exp}_x (-p/2))G_\lb(x,p)\de p,
\end{equation}
where $\text{exp}_x:T_x\Omega_\lb\to\Omega_\lb$ is the exponential map and $G_\lb(x,p)$ is a suitable weight function ensuring that $W^\lb$ is well defined. In order to guarantee that $\int_{T_x^*\Omega_\lb} W^\lb_{u,v}(x,\xi)\de\xi= \rho^2(x) u(x)\overline{v}(x)$, we only need to require that $G_\lb(x,0)=1$, for each $x\in\Omega_\lb$. Therefore, under the latter condition, we have that
$$
\int_{\Si_\lb} W^\lb_{u_\lb,v_\lb}\zeta_\lb=\langle u_\lb,v_\lb\rangle.
$$
As we mentioned in the beginning of this section, Landsman proposed a quantization with Wigner transform of the previous form ~\cite{La} II.3 (or his article~\cite{LaYM}). Let $\nu:T\Omega_\lb\to\Omega_\lb\times\Omega_\lb$ given by 
$$
\nu(x,p)=(\text{exp}_x(p/2),\text{exp}_x(-p/2)).
$$
Then, $\nu$ defines a diffeomorphism on a neighbourhood $V$ of the zero section (identified with $\Omega_\lb$) . Let 
$\kappa\in C^\infty(T\Omega_\lb)$ such that: 
\begin{enumerate}
\item[a)] $\kappa=1$ in a neighbourhood $V_0$, with $\Omega_\lb\subseteq V_0\subseteq V$.
\item[b)] $\kappa$ has support on $V$.
\item[c)] $\kappa(x,p)=\kappa(x,-p)$. 
\end{enumerate}
Then, Landsman's Wigner transform is obtained by taking $G_\lb=\rho^{-2}\kappa D(\nu)$. The corresponding Weyl-Landsman quantization (definition 3.4.4.~\cite{La}) has a number of interesting properties; for instance, it defines a strict deformation quantization on certain Poisson subalgebra $C^\infty_{PW}(T^*\Omega_\lb)$  (theorem 3.5.1~\cite{La}).

\begin{Remark}
{\rm
Let $\Op^d$ be the decomposable Weyl quantization defined by the field of Wigner transforms given by Eq.\eqref{wign}. Even if we do not have CQR on every constant of motion, we do expect to have SCQR, because $\tilde W^\lb$ can be interpreted as a sort of linearization of $W^\lb$. However, we have not developed yet the techniques required to prove such claim, and we shall leave this problem for a future work. 
}
\end{Remark}
 
Despite $\Op^d$ does not coincide with $\Op$ on {\it every} constant of motion, we still can have CQR on some important cases, and we shall finish this section showing some examples of this. It is clear that if 
\begin{equation}\label{crit}
\int_{\Si_\lb} f_\lb\tilde W^\lb(u_\lb,v_\lb)\zeta_\lb=\int_{\Si_\lb} f_\lb W^\lb(u,v)\zeta_\lb
\end{equation}
for each $u,v\in L^2(\R^n)$, then quantization and reduction commutes on $f$. For instance, if $f(x,\xi)=\varphi(x)$, where $\varphi\in C^\infty(\R^n)$, then $f$ is a constant of motion. Moreover,

$$
\int_{\Si_\lb} f_\lb\tilde W^\lb(u,v)\zeta_\lb=\int_{\Omega_\lb}\varphi(x)\int_{T_x^*\Omega_\lb} W^\lb_{u,v}(x,\xi)\de\xi=\int_{\Si_\lb} f_\lb\tilde W^\lb(u,v)\zeta_\lb.
$$
Therefore $\Op(f)=\Op^d(f)=\varphi(Q)$. In particular, $\Op(a\circ J)=\Op^d(a\circ J)=a(H_1,\cdots,H_k)$, for any $a\in C^\infty(\R^k)$. 

We shall also consider observables of the form $f=J_Y$, where $Y$ is a vector field on a manifold $\Omega$ and $J_Y\in C^\infty(T^*\Omega)$ is given by $J_Y(x,\xi)=\langle Y(x),\xi\rangle$. Such classical observables play an important role in Hamiltonian mechanics and they have many interesting properties. For instance, the Hamiltonian flow of $J_Y$ coincides with lift of the flow of $Y$ from $\Omega$ to $T^*\Omega$. Moreover, the flow of $Y$ induces an strongly continuous unitary one parameter group on $L^2(\Omega)$, and Weyl-Landsman quantization $\Op^\Omega$ (in particular, Weyl quantization, when $\Omega=\R^n$) maps $J_Y$ into the corresponding infinitesimal generator (propositions 3.6.1. and 3.6.2. in \cite{La}). More precisely, we have that
$$
\Op^\Omega(J_Y)=-i(Y+\frac{1}{2}\text{div}(Y)),
$$ 
where $\text{div}(Y)$ is the divergence of $Y$.
We shall prove that, if we define $\Op^d$ using Landsman's Wigner transform on each fiber, then quantization and reduction commute on $J_Y$, whenever $Y$ is tangent to each $\Omega_\lb$. 
\begin{Lemma}\label{div}
Let $Y$ be a vector field on $\R^n$ tangent to each $\Omega_\lb$ and denote by $Y^\lb$ the restriction of $Y$ to $\Omega_\lb$, for each $\lb\in\tilde J(\R^n)\backslash\mathcal I$. Then, for
each $z\in\Omega_{\lb}$,
$$
div(Y)(z)=div(Y^{\lb})(z)+\rho^{-2} (z)Y(\rho^2)(z)
=div(Y^{\lb})(z)+2\rho^{-1}(z)Y(\rho)(z).
$$
\end{Lemma}

\begin{proof}
The last equality follows from the chain rule. Let us denote by $\omega$ the volume form of $\R^n$. By lemma \ref{mainL}, there is a $n-k$ volume form $\mu$ such that $\omega=\mu\wedge \de\phi_1\wedge\cdots \wedge\de\phi_k$. Notice that coarea formula implies that the restriction of such $\mu$ to any $\Omega_\lb$ coincides with $\rho^2\eta_\lb$. Moreover,
$$
i_Y(\omega)=i_Y(\mu)\wedge \de\phi_1\wedge\cdots \wedge\phi_k + \mu \wedge i_Y(\de\phi_1\wedge\cdots \wedge\phi_k)=i_Y(\mu)\wedge \de\phi_1\wedge\cdots \wedge\phi_k. 
$$
Therefore, 
$$
\text{div}(Y)\mu\wedge \de\phi_1\wedge\cdots \wedge\phi_k=\text{div}(Y)\omega=d[i_Y(\mu)]\wedge \de\phi_1\wedge\cdots \wedge\phi_k.
$$
Restricting to each $\Omega_\lb$, we obtain that
$$
\text{div}(Y)\rho^2\eta_\lb=d[i_{Y^\lb}(\rho^2\eta_\lb)]=\text{div}(\rho^2 Y^\lb)\eta_\lb=[Y(\rho^2)+\rho^2\text{div}(Y^\lb)]\eta_\lb.
$$
\end{proof}
  
\begin{proposition}\label{ecqr}
Let $Y$ be a complete vector field on $\R^n$ tangent to each $\Omega_\lb$. Also, let $\Op^d$ be the decomposable Weyl quantization defined using Landsman's Wigner transform on each fiber. Quantization and reduction commutes on $J_Y$.
\end{proposition}
\begin{proof}
We know that the right hand side of \eqref{crit} coincides with 
$$
\langle -i(Y^\lb+\frac{1}{2}\text{div}(Y^\lb))\rho u,\rho v\rangle=-i\langle \rho(Y+\rho^{-1}Y(\rho)+\frac{1}{2}\text{div}(Y^\lb)) u,\rho v\rangle=-i\langle \rho^2(Y+\frac{1}{2}\text{div}(Y)) u, v\rangle
$$
Since $ -i(Y^\lb+\frac{1}{2}\text{div}(Y^\lb))$ is a selfadjoint operator, we also have that
$$
\langle -i(Y^\lb+\frac{1}{2}\text{div}(Y^\lb))\rho u,\rho v\rangle=i\langle  u, \rho^2(Y+\frac{1}{2}\text{div}(Y)) v\rangle
$$
On the left hand side of \eqref{crit}, using integration by parts and the Fourier inversion formula, we have that
$$
\int_{\Si_\lb}f_\lb \tilde W^\lb_{u,v}\zeta_\lb=\int_{\hat\Si_\lb}f \tilde W_{u,v}\gamma_\lb=\int_{\Omega_\lb}\rho^2(x)\int_{\R^n}\sum_{j=1}^n Y_j(x)\xi_j\int_{\R^n}e^{-i\xi\cdot p}u(x+\frac{p}{2})\overline{v}(x-\frac{p}{2})\de p\de\xi\de\eta_\lb(x)=
$$
$$
i\int_{\Omega_\lb}\rho^2(x)\sum_{j=1}^n Y_j(x)\int_{\R^n}\int_{\R^n}\frac{\partial}{\partial p_j}[e^{-i\xi\cdot p}]u(x+\frac{p}{2})\overline{v}(x-\frac{p}{2})\de p\de\xi\de\eta_\lb(x)=
$$
$$
-\frac{i}{2}\int_{\Omega_\lb}\rho^2(x)\sum_{j=1}^n Y_j(x)(\frac{\partial u}{\partial x_j}(x)\overline{v}(x)-u(x)\frac{\partial \overline{v}}{\partial x_j}(x))\de\eta_\lb(x)=
$$
$$
-\frac{i}{2}\int_{\Omega_\lb}\rho^2(x) Yu(x)\overline{v}(x)-u(x)Y\overline{v}(x)\de\eta_\lb(x)=-i(\langle \rho^2(Y+\frac{1}{2}\text{div}(Y)) u, v\rangle.
$$

\end{proof}
The criteria given by equation \eqref{crit} is not be a necessary condition to obtain CQR for a given contant of motion $f$. We necessary and sufficient condition is the identity obtained if we integrate both sides of \eqref{crit}. 

An interesting example on which we can apply our construction is the following. Take $k=1$ and $\phi(x)=||x||^2$. Then $\Omega_\lb=\mathbb{S}^{n-1}_{\sqrt{\lb}}$ is the $(n-1)$-dimensional sphere of radius $\sqrt{\lb}$. The canonical action of the orthogonal group $O(n)$ on $\R^n$ induces a representation of the Lie algebra $\mathfrak{so}(n)$ on the Lie algebra of vector fields on $\R^n$, and the vector fields in the range of that representation are complete and tangent to each sphere (so, we can apply the previous proposition). We can also construct many more constants of motions using the latter action. For instance, the pull back of the canonically induced moment map $M:\R^{2n}\to \mathfrak{so}(n)^*$ allows to represent $C^\infty(\mathfrak{so}(n)^*)$ as classical constants of motion. Moreover, the induced unitary representation of $O(n)$ on $L^2(\R^n)$ allows to represent as quantum constants of motion the universal enveloping algebra of $ U(\mathfrak{so}(n))$ (up to requiring selfadjointness of the corresponding operators). Note that, we can also embed $ U(\mathfrak{so}(n))$ in $C^\infty(\mathfrak{so}(n)^*)$ (and therefore in $C^\infty(\R^{2n})$).

Finally, notice that for the latter example, it can be shown that Weyl quantization preserves constants of motion (for instance, using the metaplectic representation). It would be interesting to compare $\Op$ and $\Op^d$ on $M^*[C^\infty(\mathfrak{so}(n)^*)]$, or at least on $ U(\mathfrak{so}(n))$ (after embed it). We expect that both quantizations coincides with the induced representation of $\mathcal U(\mathfrak{so}(n))$ as operators on $L^2(\R^n)$ discussed above. We shall leave this problem open for a future work.

\section{Other forms of quantization and related problems.}\label{OQ}

In this section we are going to explain briefly how our problem has been
approached from the perspective of geometric quantization, what are the differences with our approach, and what could be done within the context of deformation quantization. 

Let us consider a more general framework than the one used so far for the classical side of our problem: Let $\Si$ be a symplectic manifold, $G$ a Lie group acting by symplectomorphisms on $\Si$ and $J:\Si\to \mathfrak g^*$ a covariant moment map, where $\mathfrak{g}$ is the Lie algebra of $G$ and $\mathfrak{g}^*$ is its dual. For simplicity, we will assume that $J$ is regular and the action is free and proper on each $\hat\Si_\lb:=J^{-1}(\lb)$, where $\lb$ is a coadjoint orbit on $\mathfrak{g}^*$. Thus, the orbit space $\Si/G$ is a Poisson manifold and the quotient map $\pi:\Si\to \Si/G$ is a Poisson map. Morever, the symplectic leaves in $\Si/G$ coincide with orbit spaces $\Si_\lb:=J^{-1}(\lb)/G$ (with canonical symplectic structure given by Marsden-Weinstein-Meyer reduction). In such framework, a constant of motion is a function $a\in C^\infty(\Si)$ such that $a(g\cdot \si)=a(\si)$ for every $\si\in\Si$ and $g\in G$. The algebra of all constants of motion $\A$ can be identified with $C^\infty(\Si/G)$. 

Note that the case $G=\R^k$, can only be achieved 
with $J=(h_1,\cdots,h_k)$, where each $h_j\in C^\infty(\Si)$ and they pairwise Poisson commute. See further details in the introduction.

Geometric quantization's goal is to construct the Hilbert space and the maps required in the definition of canonical quantization using geometric data only: a Hermite line bundle over the symplectic manifold $\Si$, a suitable connection on such bundle, a polarization and the metaplectic correction \cite{S,K,W}. Such ingredients are called a prequantization data. 

It turns out that, when the prequantization data is also $G$-invariant, it simultaneously induces a unitary group representation $U$ of $G$ on the emerging Hilbert space $\H$ and a prequantization data on each $\Si_\lb$.
In particular, we can use the geometric quantization mechanism once again and build a family of Hilbert spaces $\H_\lb$. Contrary, in our context we did not construct the Hilbert spaces $\H(\lb)$, they were given abstractly within the framework of spectral theory and we justified why they should be the quantum counterpart of $\Si_\lb$ (see the introduction).

A natural question arises: Is there a procedure to obtain $\H_\lb$ from $\H$ similarly to the way how $\Si_\lb$ is constructed from $\Si$? . A positive answer was given by Guillemin and Sternberg in the celebrated article \cite{GS}. They proved that, if $G$ and $\Si$ are compact then $\H_0$ coincides with $\H^0:=\{\psi\in\H/U_g(\psi)=\psi,\forall g\in G\}$; moreover, that result lead to explicit expressions for each $\H_\lb$ as well. Other important cases were obtained in \cite{Go}.

Such results are sometimes called results of the "commutation of quantization and reduction" type. None of those results cover the case $G=\R^k$ and $\Si=\R^{2n}$; the main reason why that happen is that most of the approaches in the literature end up proposing $\H^0$ as the Hilbert space $\H_0$, and that space could be trivial if $\H$ is infinite dimensional or $G$ is not compact. For instance, $\H^0$ is trivial when $k=1$, $G=\R$, and the action is given by the Hamiltonian flow of either a position coordinate Hamiltonian or a momentum coordinate Hamiltonian or the free Hamiltonian. Following the ideas motivating this article (see the introduction), we propose a different approach: Recall that there is a unitary representation $U$ of $G$ on $\H$ and consider the decomposition of $U$ as a direct integral over the space of irreducible representations of $G$ (this can be done for a large class of groups called type I groups \cite{Mau}). Notice that, under certain conditions, the space of irreducible representations of $G$ can be identified with the set of coadjoint orbits of $\mathfrak{g}^*$ (this is another famous result in geometric quantization \cite{K}). We conjecture that the Hilbert space $\H_\lb$ coming from $\Si_\lb$ coincides with the fiber Hilbert space required in the decomposition of $U$ corresponding to $\lb$. Notice that our $\H(\lb)$ can be obtained in that way if we take $U(t)=e^{i(t_1H_1+\cdots+t_k H_k)}$. 

We shall also notice that the dimensional restriction that we imposed in our framework, i.e. $k<n$, might not be necessary in geometric quantization. For example, if $h_j(x,\xi)=\xi_j$ with $j=1,\cdots, n$, then the composition of the corresponding Hamiltonian flows coincides with the translation action of $\R^n$ on itself lifted to $T^*\R^n$. Such example has been studied in geometric quantization even replacing $\R^n$ by an arbitrary Lie group $G$ (for instance, see corollary 3.13. in \cite{HM}).

We now reflect on how our ideas relate to deformation quantization. Recall that in deformation quantization the goal is to study the star products on a given Poisson algebra. The relation with canonical quantization comes from the so called Moyal star product, which we denote by $\star_M$. Moyal star product satisfies $\Op_\hb(f\star_M g)=\Op_\hb(f)\Op_\hb(g)$, where $\Op_\hb$ denotes Weyl quantization \cite{BFFLS, EGV}. Since $\A=C^\infty(\Si/G)$ and $\Si/G$ is a Poisson manifold, the existence of a star product for the algebra of constants of motion is guaranteed by Kontsevich's formality theorem \cite{K}. However, in the light of the results of this article, we expect that we can construct a new star product $\star_d$ defined on $\A$ such that $\Op^d(f\star^d g)=\Op^d(f)\Op^d(g)$. We also expect that such star product satisfies the identity $f\star^d h_j=h_j\star^d f$.
Moreover, we expect that such star product admits an expression of the form $f\star_d g=\int f_\lb\star_\lb g_\lb \de\lb$, where $\star_\lb$ is a star product on $C^\infty(\Si_\lb)$. Such expression might be generalized to more general context (with $G\neq\R^k$) and without requiring the construction of $\Op^d_\hb$. 

Concerning the preservation of constants of motion and commutation of quantization and reduction, we can wonder when a given star product on $C^\infty(\Si)$ can be restricted to $\A$ and when such restriction coincide with the star product $\star^d$. 





\end{document}